\documentclass[letter,11pt]{article}

\usepackage[utf8]{inputenc}
\usepackage{amsthm, amsmath, amssymb,xcolor,thm-restate,environ,comment,url, hyperref}
\usepackage{algo}
\usepackage{caption}
\usepackage[margin=1in]{geometry}
\usepackage{enumitem}
\usepackage{graphicx}
\usepackage{color}

\newcommand*{\email}[1]{\texttt{#1}}

\newlist{casesp}{enumerate}{4} 
\setlist[casesp]{align=left, 
                 listparindent=\parindent, 
                 parsep=\parskip, 
                 font=\normalfont\bfseries, 
                 leftmargin=0pt, 
                 labelwidth=0pt, 
                 itemindent=.4em,labelsep=.4em, 
                 partopsep=0pt, 
                 }
\setlist[casesp,1]{label=Case~\arabic*:,ref=\arabic*}
\setlist[casesp,2]{label=Case~\thecasespi.\alph*:,ref=\thecasespi.\alph*}
\setlist[casesp,3]{label=Case~\thecasespii.\roman*:,ref=\thecasespii.\roman*}
\setlist[casesp,4]{label=Case~\thecasespiii.\Alph*:,ref=\thecasespiii.\Alph*}

\def\R{\mathbb{R}}
\def\N{\mathbb{N}}
\def\Z{\mathbb{Z}}
\def\complement#1{\overline{#1}}
\def\set#1{\{ #1 \}}

\newtheorem{theorem}{Theorem}[section]

\newtheorem{proposition}[theorem]{Proposition}
\theoremstyle{definition}
\newtheorem{definition}{Definition}[section]
\newtheorem{claim}{Claim}[section]

\NewEnviron{problem}[1]{%
	\begin{center}\fbox{\parbox{5in}{%
				{\centering\scshape #1\par}%
				\parskip=1ex
				\everypar{\hangindent=1em}%
				\BODY
}}\end{center}}

\renewcommand{\figurename}{Algorithm}

\title{Multicriteria Cuts and Size-Constrained \texorpdfstring{$k$-Cuts}{k-Cuts} in Hypergraphs}
    \author{Calvin Beideman\thanks{University of Illinois, Urbana-Champaign. Email: \email{\{calvinb2, karthe\}@illinois.edu}. Supported in part by NSF CCF-1907937 and CCF-1814613.}
	\and 
	Karthekeyan Chandrasekaran\footnotemark[1]
	\and
	Chao Xu\thanks{The Voleon Group. Email: \email{the.chao.xu@gmail.com}.}
}

\date{}

\begin{document}
	
	\maketitle
	
	\begin{abstract}
    
    We address counting and optimization variants of multicriteria global min-cut and size-constrained min-$k$-cut in hypergraphs. 
    \begin{enumerate}
        \item For an $r$-rank $n$-vertex hypergraph endowed with $t$ hyperedge-cost functions, we show that the number of multiobjective min-cuts is $O(r2^{tr}n^{3t-1})$. In particular, this shows that the number of parametric min-cuts in constant rank hypergraphs for a constant number of criteria is strongly polynomial, thus resolving an open question by Aissi, Mahjoub, McCormick, and Queyranne \cite{aissi2015strongly}. In addition, we give randomized algorithms to enumerate all multiobjective min-cuts and all pareto-optimal cuts in strongly polynomial-time. 
        
        \item We also address node-budgeted multiobjective min-cuts: For an $n$-vertex hypergraph endowed with $t$ vertex-weight functions, we show that the number of node-budgeted multiobjective min-cuts is $O(r2^{r}n^{t+2})$, where $r$ is the rank of the hypergraph, and the number of node-budgeted $b$-multiobjective min-cuts for a fixed budget-vector $b\in \R^t_+$ is $O(n^2)$. 
    
        \item We show that min-$k$-cut in hypergraphs subject to constant lower bounds on part sizes is solvable in polynomial-time for constant $k$, thus resolving an open problem posed by Queyranne \cite{queyranne2020sizeconstrained}. Our technique also shows that the number of optimal solutions is polynomial. 
    \end{enumerate}
    All of our results build on the random contraction approach of Karger \cite{Kar93}. 
    Our techniques illustrate the versatility of the random contraction approach to address counting and algorithmic problems concerning multiobjective min-cuts and size-constrained $k$-cuts in hypergraphs. 
\end{abstract}

	\section{Introduction} \label{sec:intro}
Cuts and partitioning play a central role in combinatorial optimization and have numerous theoretical as well as practical applications. 
We consider multicriteria cut problems in hypergraphs. Let $G=(V,E)$ be an $n$-vertex hypergraph and $c_1, \ldots, c_t: E\rightarrow \Z_+$ be $t$ non-negative hyperedge-cost functions, where $t$ is a constant. 
The cost of a subset $F$ of hyperedges under criterion $i\in [t]$ is $c_i(F):=\sum_{e\in F}c_i(e)$. 
For a positive integer $k$, a subset of hyperedges that crosses a $k$-partition $(U_1, \ldots, U_k)$ of the vertex set is said to be a $k$-cut. We refer to a $2$-cut simply as a cut. We recall that the rank of a hypergraph $G$ is the size of the largest hyperedge in $G$ (the rank of a graph is $2$). 

Since we have several criteria, there may not be a single cut that is best for all criteria. 
In multicriteria optimization, there are three important notions to measure the quality of a cut: (i) parametric min-cuts, (ii) pareto-optimal cuts, and (iii) multiobjective min-cuts. 
A cut $F$ is a \emph{parametric min-cut} if there exist positive multipliers $\mu_1, \ldots, \mu_t\in \R_+$ such that $F$ is a min-cut in the hypergraph $G$ with hyperedge costs given by $c_{\mu}(e):=\sum_{i=1}^t \mu_i c_i(e)$ for all $e\in E$. 
A cut $F$ \emph{dominates} another cut $F'$ if $c_i(F)\le c_i(F')$ for every $i\in [t]$ and there exists $i\in [t]$ such that $c_i(F)<c_i(F')$. A cut $F$ is \emph{pareto-optimal} if it is not dominated by any other cut. 
For a budget-vector $b\in \R^{t-1}_+$, a cut $F$ is a \emph{$b$-multiobjective min-cut} if $c_i(F)\le b_i$ for every $i\in [t-1]$ and $c_t(F)$ is minimum subject to these constraints. A cut $F$ is a \emph{multiobjective min-cut} if there exists a non-negative budget-vector $b\in \R^{t-1}_+$ for which $F$ is a $b$-multiobjective min-cut. 
These three notions satisfy the following relationship with the containment being possibly strict (see Appendix \ref{sec:comparison} for a proof):
\begin{align} \label{containment-relationships}
\text{Parametric min-cuts} &\subseteq \text{Pareto-optimal cuts} \subseteq \text{Multiobjective
min-cuts}.
\end{align}
	
There is also a natural notion of min-cuts under node-weighted budget constraints. 
Let $w_1, \ldots, w_t: V\rightarrow \R_+$ be vertex-weight functions and $c: E\rightarrow \R_+$ be a hyperedge-cost function. For a budget-vector $b\in \R^t_+$, a subset $F\subseteq E$  of hyperedges is a \emph{node-budgeted $b$-multiobjective min-cut} if 
$F=\delta(U)$ for some subset $\emptyset\neq U\subsetneq V$ with $\sum_{u\in U}w_i(u)\le b_i$ for all $i\in [t]$  and $c(F)$ is minimum among all such subsets of $E$. A cut $F$ is a \emph{node-budgeted multiobjective min-cut} if there exists a non-negative budget-vector $b$ for which $F$ is a node-budgeted $b$-multiobjective min-cut. 
In this work, we address the following natural questions concerning multiobjective min-cuts and min-$k$-cuts:
\begin{enumerate}
    \item Multiobjective min-cuts: Is the number of multiobjective min-cuts at most strongly polynomial?
    \item Node-budgeted multiobjective min-cuts: Is the number of node-budgeted multiobjective min-cuts at most strongly polynomial? 
    \item Size-constrained min-$k$-cut: 
    For fixed positive integers $k$ and $s_1, \ldots, s_k$ (all constants), a vertex-weight function $w:V\rightarrow \Z_+$, and a hyperedge-cost function $c:E\rightarrow \R_+$,
    can we compute a $k$-cut $F$ with minimum $c(F)$ subject to the constraint that $F$ is the set of hyperedges crossing some $k$-partition $(U_1, \ldots, U_k)$ of $V$ where $\sum_{u\in U_i}w(u)\ge s_i$ for every $i\in [k]$ in polynomial-time?
    Is the number of optimal solutions strongly polynomial?
\end{enumerate}

\paragraph{Previous work.} 
For single criterion, a classic result of Dinitz, Karzanov, and Lomonosov \cite{DKL76} shows that the number of min-cuts in an $n$-vertex graph is $O(n^2)$ (also see Karger \cite{Kar93}). The same upper bound was shown to hold for constant-rank hypergraphs by Kogan and Krauthgamer \cite{Kogan2015sketching} and for arbitrary-rank hypergraphs by Chekuri and Xu \cite{CX18} and by Ghaffari, Karger, and Panigrahi \cite{GKP17} via completely different techniques. 
For $t=2$ criteria in graphs, Mulmuley \cite{Mul99} showed an $O(n^{19})$ upper bound on the number of parametric min-cuts. 
For $t$ criteria in constant-rank hypergraphs, Aissi, Mahjoub, McCormick, and Queyranne \cite{aissi2015strongly} showed that the number of parametric min-cuts is $\tilde{O}(m^{t}n^2)$, where $m$ is the number of hyperedges, using the fact that the number of approximate min-cuts in constant-rank hypergraphs is polynomial. Karger \cite{Karger2016minimish} improved this bound to $O(n^{t+1})$ by a clever and subtle argument based on his random contraction algorithm; we will describe his argument later. Karger also constructed a graph that exhibited $\Omega(n^{t/2})$ parametric min-cuts.  

Armon and Zwick \cite{AZ06} showed that all pareto-optimal cuts in graphs can be enumerated in pseudo-polynomial time. 
For $t=2$ criteria in constant-rank hypergraphs, Aissi et al \cite{aissi2015strongly} showed an upper bound of $\tilde{O}(n^5)$ on the number of pareto-optimal cuts---this was the first result showing a  strongly polynomial upper bound. Aissi et al raised the question of whether the number of pareto-optimal cuts is strongly polynomial for a constant number $t$ of criteria in constant-rank hypergraphs (or even in graphs). Note that, by containment relationship (\ref{containment-relationships}), answering our first question affirmatively would also answer their open question. 
On a related note, Aissi, Mahjoub, and Ravi \cite{Aissi2017multiobjective} designed a random contraction based algorithm to solve the $b$-multiobjective min-cut problem in graphs. The correctness analysis of their algorithm also implies that the number of $b$-multiobjective min-cuts in graphs for a fixed budget-vector $b\in \R^{t-1}_+$ is $O(n^{2t})$. We emphasize the subtle, but important, distinction between the number of $b$-multiobjective min-cuts for a fixed budget-vector $b$ and the number of multiobjective min-cuts. 

Node-budgeted multiobjective min-cut has a rich literature extending nicely to submodular functions.
For graphs, Armon and Zwick \cite{AZ06} gave a polynomial-time algorithm to find a minimum valued cut with at most $b$ vertices in the smaller side. Goemans and Soto \cite{GS13} addressed the more general problem of minimizing a symmetric submodular function $f:2^V\rightarrow \R$ over a downward-closed family $\mathcal{I}$. Recall that the hypergraph cut function is symmetric submodular and the family of vertex subsets satisfying node-weighted budget constraints is in fact downward-closed. Goemans and Soto extended Queyranne's submodular minimization algorithm to enumerate all the $O(n)$ minimal minimizers in $\mathcal{I}$ using $O(n^3)$ oracle calls to the function $f$ and the family $\mathcal{I}$. Their result implies that the number of minimal minimizers is $O(n)$, but 
it is straightforward to see that the total number of minimizers could be exponential. 
For the special case of node-budgeted multiobjective min-cuts in graphs, Aissi, Mahjoub, and Ravi \cite{Aissi2017multiobjective} gave a faster algorithm than that of Goemans and Soto---their algorithm is based on random contraction, runs in $\tilde{O}(n^2)$-time, and shows that the number node-budgeted  $b$-multiobjective min-cuts in graphs for a fixed budget-vector $b\in \R^t_+$ is $O(n^2)$. 

For size-constrained min-$k$-cut, if we allow arbitrary sizes (i.e., arbitrary lower bounds), then the problem becomes NP-hard even for $k=2$ as it captures the well-studied min-bisection problem in graphs. If we consider constant sizes but arbitrary $k$, then the problem is again NP-hard in graphs \cite{GH94}. So, our focus is on constant $k$ and constant sizes. 
Guinez and Queyranne \cite{queyranne2020sizeconstrained} raised size-constrained min-$k$-cut with unit vertex-weights as a sub-problem towards resolving the complexity of the submodular $k$-partitioning problem. In submodular $k$-partitioning, we are given a submodular function $f:2^{V}\rightarrow \R$ (by value oracle) and a fixed constant integer $k$ (e.g., $k=2,3,4,5, \ldots$) and the goal is to find a  $k$-partition $(U_1, \ldots, U_k)$ of the ground set $V$ so as to minimize $\sum_{i=1}^k f(U_i)$. The complexity of even special cases of this problem are open: e.g., if the submodular function $f$ is the cut function of a given hypergraph, then its complexity is unknown.\footnote{We note that if the submodular function $f$ is the cut function of a given hypergraph, then the submodular $k$-partition problem is not identical to hypergraph $k$-cut as the two objectives are different. 

However, if the submodular function is the cut function of a given graph, then the submodular $k$-partition problem coincides with the graph $k$-cut problem which is solvable in polynomial-time.} Guinez and Queyranne showed surprisingly strong non-crossing properties between optimum solutions to size-constrained $(k-1)$-partitioning (constant size lower bounds on the parts) and optimum solutions to $k$-partitioning. This motivated them to study the size-constrained min-$k$-cut problem in hypergraphs for unit vertex-weights as a special case. They showed that size-constrained min-$k$-cut for unit vertex-weights is solvable in polynomial-time in constant-rank hypergraphs (with exponential run-time dependence on the rank) and mention the open problem of designing an algorithm for it in arbitrary-rank hypergraphs. 
The size-constrained min-$k$-cut problem for unit sizes (i.e., all size lower-bounds $s_1, \ldots, s_k$ are equal to one) is known as the hypergraph $k$-cut problem. The hypergraph $k$-cut problem was shown to admit a polynomial-time algorithm only recently \cite{chandrasekararen2018hypergraph} via a non-uniform random contraction algorithm. 

\subsection{Our Contributions}
Our high-level contribution is in showing the versatility of the random contraction technique to address 
algorithmic and counting problems concerning multiobjective min-cuts and size-constrained min-$k$-cuts in hypergraphs. 
All of our results build on the random contraction technique with additional insights.

Our first result is a strongly polynomial upper bound on the number of multiobjective min-cuts in constant-rank hypergraphs. 
\begin{restatable}{theorem}{numMultiobjective}\label{thm:num_multi}
The number of multiobjective min-cuts in an $r$-rank, $n$-vertex hypergraph $G$ with $t$ hyperedge-cost functions is $O(r2^{rt}n^{3t-1})$. 
\end{restatable}

We emphasize that our upper bound is over all possible non-negative budget-vectors (in contrast to the number of $b$-multiobjective min-cuts for a fixed budget-vector $b$).
Theorem \ref{thm:num_multi} and Proposition \ref{containment-relationships} imply that the number of pareto-optimal cuts in constant-rank hypergraphs is $O(n^{3t-1})$ and hence, is strongly polynomial for constant number of criteria. This answers the main open question posed by Aissi, Mahjoub, McCormick, and Queyranne \cite{aissi2015strongly}. We also design randomized polynomial time algorithms to enumerate all multiobjective min-cuts and all pareto-optimal cuts in constant-rank hypergraphs (see Section \ref{sec:verifying}). 
Independent of our work, Rico Zenklusen has also shown Theorem \ref{thm:num_multi}. We learned after submission of this work that his approach also leads to \emph{deterministic} polynomial time algorithms to enumerate all multiobjective min-cuts and all pareto-optimal cuts in constant-rank hypergraphs.

Given the upper bound in Theorem \ref{thm:num_multi}, a discussion on the lower bound is in order. We recall that Karger \cite{Karger2016minimish} constructed a graph with $t$ edge-cost functions that exhibited $\Omega(n^{t/2})$ parametric min-cuts. This is also a lower bound on the number of pareto-optimal cuts and multiobjective min-cuts by (\ref{containment-relationships}). We improve this lower bound for pareto-optimal cuts by constructing a graph with $t$ edge-cost functions that exhibits $\Omega(n^{t})$ pareto-optimal cuts (see Section \ref{sec:LowerBounds}). Our instance also exhibits the same lower bound on the number of $b$-multiobjective min-cuts for a fixed budget-vector $b$. 

Our next result is an upper bound on the number of node-budgeted multiobjective min-cuts and node-budgeted $b$-multiobjective min-cuts. 
\begin{restatable}{theorem}{numNodeBudgetedMultiobjective}\label{thm:num_node_budgeted_multi}
\begin{enumerate}
\item The number of node-budgeted multiobjective min-cuts in an $r$-rank, $n$-vertex hypergraph with $t$ vertex-weight functions is $O(r2^{r}n^{t+2})$. 
\item For a fixed budget-vector $b\in \R^t_+$, the number of node-budgeted $b$-multiobjective min-cuts in an $n$-vertex hypergraph with $t$ vertex-weight functions is $O(n^2)$. 
\end{enumerate}
\end{restatable}

We draw the reader's attention to the distinction between the two parts in Theorem \ref{thm:num_node_budgeted_multi}. 
The first part implies that the number of node-budgeted multiobjective min-cuts is strongly polynomial in constant-rank hypergraphs for constant number of vertex-weight functions. The second part implies that the number of node-budgeted $b$-multiobjective min-cuts for any fixed budget-vector $b\in \R^t_+$ is strongly polynomial in arbitrary-rank hypergraphs for any number $t$ of vertex-weight functions. 

Our final result shows that the size-constrained min-$k$-cut problem  can be solved in polynomial time for constant $k$ and constant sizes (in arbitrary-rank hypergraphs). 
\begin{restatable}{theorem}{sizeConstrainedSuccessProb}\label{thm:size_constrained_success_prob}
Let $k\ge 2$ be a fixed positive integer and let $1\le s_1 \le s_2 \le \ldots \le s_k $ be fixed positive integers.
Let $G$ be an $n$-vertex hypergraph with hyperedge-cost function $c:E\rightarrow \R_+$. Then, there exists a polynomial-time algorithm that takes $(G,c)$ as input and returns a fixed $w$-weighted $s$-size-constrained min-$k$-cut for any choice of vertex-weight function $w:V \rightarrow \Z_+$ with probability
\[
\Omega\left(\frac{1}{n^{2\sigma_{k-1}+1}}\right),
\]
where $\sigma_{k-1}:=\sum_{i=1}^{k-1} s_i$.
\end{restatable}
Theorem \ref{thm:size_constrained_success_prob} resolves an open problem posed by Guinez and Queyranne \cite{queyranne2020sizeconstrained}. 
A structural consequence of Theorem \ref{thm:size_constrained_success_prob} is that the number of size-constrained min-$k$-cuts (over all possible node-weight functions $w:V\rightarrow \Z_+$) in a given hypergraph is polynomial for constant sizes and constant $k$. 

We refer the reader to Table \ref{fig:table-of-results} for a comparison of known results and our contributions. 

\begin{table}[htb]
\centering
\begin{tabular}{|c|c|c|c|}
    \hline
    \textbf{Problem} & \textbf{Graphs} & $r$\textbf{-rank Hypergraphs} & \textbf{Hypergraphs}\\ [0.5ex]
    \hline
    \# of Parametric Min-Cuts & {\color{gray} $O(n^{t+1})$ \cite{Karger2016minimish}} & {\color{gray} $O(2^rn^{t+1})$ \cite{Karger2016minimish}} & OPEN \\
    & {\color{gray} $\Omega(n^{t/2})$ \cite{Karger2016minimish}} && \\
    \hline
    \# of Pareto-Optimal Cuts & {\color{gray} $\tilde{O}(n^5)$ for $t=2$ \cite{aissi2015strongly}} & ${\color{black} O(r2^{rt}n^{3t-1})}$ [Thm \ref{thm:num_multi}] & OPEN\\
    & ${\color{black} O(n^{3t-1})}$ [Thm \ref{thm:num_multi}] && \\
    & ${\color{black} \Omega(n^t)}$ [Thm \ref{thm:pareto_lower_bound}] && \\ 
    \hline
    \# of $b$-Multiobjective & {\color{gray} $O(n^{2t})$ \cite{Aissi2017multiobjective}} & ${\color{black} O(r2^{rt}n^{2t}) }$ [Thm \ref{thm:suc_prob}] & OPEN \\
    Min-Cuts& ${\color{black} \Omega(n^t) }$ [Thm \ref{thm:pareto_lower_bound}] && \\
    \hline
    \# of Multiobjective & ${\color{black} O(n^{3t-1})}$ [Thm \ref{thm:num_multi}] & ${\color{black} O(r2^{rt}n^{3t-1}) }$  [Thm \ref{thm:num_multi}] & OPEN \\
    Min-Cuts& $\Omega(n^t)$ [Thm \ref{thm:pareto_lower_bound}] &&\\
    \hline
    \# of Node-Budgeted  & {\color{gray} $O(n^2)$ \cite{Aissi2017multiobjective}} & ${\color{black} O(n^2)} $ [Thm \ref{thm:num_node_budgeted_multi}] & ${\color{black} O(n^2) }$ [Thm \ref{thm:num_node_budgeted_multi}] \\
    $b$-Multiobjective Min-Cuts &&&\\
    \hline
    \# of Node-Budgeted  & ${\color{black} O(n^{t+2}) }$ [Thm \ref{thm:num_node_budgeted_multi}] & ${\color{black} O(r2^{r}n^{t+2}) }$ [Thm \ref{thm:num_node_budgeted_multi}]& OPEN \\
    Multiobjective Min-Cuts &&& \\
    \hline
    Node-Weighted  & {\color{gray} Poly-time \cite{queyranne2020sizeconstrained}} & {\color{gray} Poly-time  \cite{queyranne2020sizeconstrained}} & Poly-time  \\
    $s$-Size-Constrained $k$-cut &&&[Thm \ref{thm:size_constrained_success_prob}]\\
    (const. $k$ and const. $s\in \Z^{k}$) &&&\\
    \hline
\end{tabular}

\caption{Text in gray refers to known results while text in black illustrates results from this work. Here, $t$ denotes the number of criteria (i.e., the number of hyperedge-cost/vertex-weight functions), $r$ denotes the rank of the hypergraph, and $n$ denotes the number of vertices.}
\label{fig:table-of-results}
\end{table}

\subsection{Technical Overview}
As mentioned earlier, all our results build on the random contraction technique introduced by Karger \cite{Kar93} to solve the global min-cut problem in graphs. Here, a \emph{uniform} random edge of the graph is contracted in each step until the graph has only two nodes; the set of edges between the two nodes is returned as the cut. Karger showed that this algorithm returns a fixed global min-cut with probability $\Omega(n^{-2})$. As a consequence, the number of min-cuts in an $n$-vertex graph is $O(n^2)$. The algorithm extends naturally to $r$-rank hypergraphs, however the naive analysis will only show that the algorithm returns a fixed global min-cut with probability $\Omega(n^{-r})$. Kogan and Krauthgamer \cite{Kogan2015sketching} introduced an LP-based analysis thereby showing that the algorithm indeed succeeds with probability $\Omega(2^{-r} n^{-2})$. As a consequence, the number of global min-cuts in constant-rank hypergraphs is also $O(n^2)$. 

In a recent work, Karger observed that uniform random contraction can also be used to bound the number of parametric min-cuts in constant-rank hypergraphs. We describe his argument for graphs since two of our  theorems build on it. 
Suppose we fix the multipliers $\mu_1, \ldots, \mu_t$ in the parametric min-cut problem, then 
a fixed min $c_{\mu}$-cost cut can be obtained with probability $\Omega(n^{-2})$ by running the random contraction algorithm with respect to the edge-cost function $c_{\mu}$.
Karger suggested an alternative viewpoint of the execution of the algorithm for the edge-cost function $c_{\mu}$. 
For simplicity, we assume parallel edges instead of costs, i.e., $c_i(e)\in \set{0,1}$ for every edge $e$ and every criterion $i\in [t]$. 
Let $E_i$ be the set of edges with non-zero weight in the $i$'th criterion. 
The execution of the random-contraction algorithm wrt $c_{\mu}$ can alternatively be specified as follows: a permutation $\pi_i$ of the edges in $E_i$ for each $i\in [t]$ and an \emph{interleaving} indicating at each step whether the algorithm contracts the next edge from $\pi_1$ or $\pi_2$ or $\ldots$ or $\pi_t$. Critically, the sequences $\pi_i$ for every $i\in [t]$ can be assumed to be uniformly random. Thus, we can move all randomness upfront, namely pick a uniform random permutation $\pi_i$ for each criterion $i\in [t]$. 
Now, instead of returning one cut, we return the collection of cuts produced by contracting along all possible interleavings. This modified algorithm no longer depends on the specific multipliers $\mu_1, \ldots, \mu_t$ and hence, a parametric min-cut for any fixed choice of multipliers $\mu_1, \ldots, \mu_t$ will be in the output collection with probability at least $\Omega(n^{-2})$. 
It remains to bound the number of interleavings since that determines the number of cuts in the returned collection: the crucial observation here is that the number of interesting interleavings is only $n^{t-1}$. This is because interleaved contractions produce the same final graph as performing a certain number of contractions according to $\pi_1$ (until the number of vertices is $n_1$ say), then a certain number of contractions based on $\pi_2$ (until the number of vertices is $n_2$ say), and so on. So, the order of contractions becomes irrelevant and only the number of vertices $n_1, \ldots, n_t$ are relevant. 
Overall, this implies that the number of parametric min-cuts is $O(n^{t+1})$. We emphasize that this \emph{interleaving argument} relies crucially on the basic random contraction algorithm picking edges to contract according to a \emph{uniform distribution} (allowing the permutations $\pi_1, \ldots, \pi_t$ to be uniform random permutations). 

Next, we describe our approach underlying the proof of Theorem \ref{thm:num_multi}, but for graphs. 
In order to bound the number of multiobjective min-cuts through the interleaving argument, we first need a random-contraction based algorithm to solve the $b$-multiobjective min-cut problem. Indeed, Aissi, Mahjoub, and Ravi \cite{Aissi2017multiobjective} designed a random-contraction based algorithm to solve the $b$-multiobjective min-cut problem in graphs. Their algorithm proceeds as follows: Let $U_0:=\emptyset$ and for each $i\in [t-1]$, let $U_i$ be the set of vertices $u\in V-\cup_{j=1}^{i-1}U_j$ for which $c_i(\delta(u))>b_i$ (known as the set of $i$-infeasible vertices), and let $U_t:=V-\cup_{j=1}^{t-1} U_j$. In each step, they pick $i\in [t]$ with probability proportional to the number of $i$-infeasible vertices (i.e., $|U_i|$) and pick a random edge $e$ among the ones incident to $U_i$ with probability proportional to $c_i(e)$, contract $e$, and repeat. Unfortunately, this algorithm does not have the \emph{uniform distribution} that is crucially necessary to apply Karger's interleaving argument. To introduce uniformity to the distribution, we modify this algorithm in two ways:
\begin{enumerate}
    \item At each step 
    we deterministically choose the criterion $i$ corresponding to the largest $U_i$ (as opposed to picking $i$ randomly with probability proportional to $|U_i|$).
	\item 
	Next, we choose a uniform random edge $e$ from among all edges in the graph with probability proportional to $c_i(e)$ (as opposed to picking an edge only from among the edges incident to $U_i$). We contract this chosen edge $e$. 
\end{enumerate}
These two features bring a \emph{uniform distribution} property to the algorithm, which in turn, allows us to apply the interleaving argument. With these two features, we show that the algorithm returns a fixed $b$-multiobjective min-cut for a fixed budget-vector $b$ with probability $\Omega(n^{-2t})$. Armed with the two features, we move all randomness upfront using the interleaving argument. As a consequence, we obtain that the total number of multiobjective min-cuts (irrespective of the choice of budget-vector $b$) is $O(n^{3t-1})$. For constant-rank hypergraphs, we perform an LP-based analysis of our algorithm for $b$-multiobjective min-cut (thus, extending Kogan and Krauthgamer's analysis) to arrive at the same success probability. The interleaving argument for constant-rank hypergraphs proceeds similarly. 

We emphasize that the interleaving argument does not extend to arbitrary-rank hypergraphs. This is because, the random contraction based algorithm that we know for arbitrary-rank hypergraphs crucially requires non-uniform contractions (the next hyperedge to contract is chosen from a distribution that depends on the current sizes of all hyperedges), so we cannot assume that the permutations $\pi_1, \pi_2, \ldots, \pi_t$ are uniformly random. Consequently, we do not even know if the number of parametric  min-cuts in a hypergraph is at most strongly polynomial. Another interesting open question here is whether the $b$-multiobjective min-cut problem in hypergraphs is solvable in polynomial-time even for $t=2$ criteria. We have arrived at hypergraph instances (with large rank) for which Aissi, Mahjoub, and Ravi's approach (as well as our modified approach) will never succeed, even with non-uniform random contractions. 

Next, we outline the proof of Theorem \ref{thm:num_node_budgeted_multi}. The approach is to again design a random-contraction algorithm that returns a fixed node-budgeted $b$-multiobjective min-cut with probability $\Omega(n^{-2})$ (in both constant-rank and arbitrary-rank hypergraphs). Such an algorithm would imply the second part of the theorem immediately while the first part would follow if we can apply an interleaving-like argument (i.e., the designed algorithm performs uniform random contractions). Our approach is essentially an extension of the approach by Goemans and Soto who suggested contracting the infeasible vertices together (a vertex $u$ is infeasible if $w_i(u)>b_i$ for some $i\in [t]$. Aissi, Mahjoub, and Ravi show that doing this additional step after each random contraction step returns a fixed node-budgeted $b$-multibjective min-cut with probability $\Omega(n^{-2})$ in graphs. Our main contribution is showing that this additional ``contracting infeasible vertices together'' step in conjunction with (1) uniform random contractions for constant-rank hypergraphs and (2) non-uniform random contractions for arbitrary hypergraphs succeeds with the required probability. Next, a naive interleaving-like argument can be applied for constant-rank hypergraphs to conclude that the number of node-budgeted multiobjective min-cuts is $O(n^{t+3})$. We improve this to $O(n^{t+2})$ with a more careful argument. 

Finally, we outline our approach for Theorem \ref{thm:size_constrained_success_prob}. 
Guinez and Queyranne address size-constrained $k$-cut in constant-rank hypergraphs for unit vertex-weights by performing uniform random contractions until the number of nodes in the hypergraph is close to $\sum_{i=1}^k s_i$ at which point they return a uniform random cut. Their success probability has exponential dependence on the rank. The key technical ingredient to bring down the exponential dependence on rank is the use of non-uniform contractions. For the special case
of unit sizes and unit vertex-weights 
(i.e., the hypergraph $k$-cut problem), Chandrasekaran, Xu, and Yu \cite{chandrasekararen2018hypergraph} introduced an explicit non-uniform distribution that leads to a success probability of $\Omega(n^{-2(k-1)})$. Our algorithm extends the non-uniform distribution to arbitrary but constant sizes (as opposed to just unit sizes), yet without depending on vertex-weights. Our analysis takes care of the vertex-weight function through weight tracking, i.e., by declaring the weight of a contracted node to be the sum of the weight of the vertices in the hyperedge being contracted. We note that our algorithm's success probability when specialized to the case of unit vertex-weights and unit sizes is weaker than the success probability of the algorithm by Chandrasekaran, Xu, and Yu (by a factor of $n$). We leave it as an open question to improve this. On the other hand, our algorithm has the added advantage that it does not even take the vertex-weight function $w$ as input and yet succeeds in returning a $w$-vertex-weighted $s$-size-constrained $k$-cut for any choice of $w$ with inverse polynomial probability. 

\paragraph{Organization.} In Section \ref{sec:edge-costs}, we bound the number of multiobjective min-cuts (and prove Theorem \ref{thm:num_multi}), give efficient algorithms to enumerate all multiobjective min-cuts and all pareto-optimal cuts, and present lower bounds on the number of pareto-optimal cuts and the number of $b$-multiobjective min-cuts. In Section \ref{sec:node_budgets}, we address node-budgeted multiobjective min-cuts and prove Theorem \ref{thm:num_node_budgeted_multi}. In Section \ref{sec:size_constrained}, we give an algorithm to solve size-constrained min-$k$-cut, thereby proving Theorem \ref{thm:size_constrained_success_prob}.

	We define the random contraction procedure that is central to all of our algorithms. Let $G= (V,E)$ be a hypergraph, and let $U \subseteq V$ be a set of vertices in $G$. We define $G$ \emph{contract} $U$, denoted $G / U$, to be a hypergraph on the vertex set $(V \setminus U) \cup \{u\}$, where $u$ is a newly introduced vertex. The hyperedges of $G/U$ are defined as follows. For each hyperedge $e \in E$ of $G$, such that $e \not\subseteq U$, $G/U$ has a corresponding hyperedge $e'$, where $e' := e$ if $e  \subseteq V \setminus U$ and $e' := (e \setminus U) \cup \{u\}$ otherwise. If $w$ is a vertex-weight function for $G$, then we will also use $w$ as a vertex-weight function for $G / U$. We define the weight of the newly introduced vertex $u$ as $w(u) := \sum_{v \in U} w(v)$.

We will need the following lemma that will be used in the analysis of two of our algorithms. We present its proof in the appendix. 

\begin{restatable}{lemma}{LPSol}\label{lem:lp_sol}
	Let $r, \gamma, n$ be positive integers with $n\ge \gamma \ge r+1 > 2$. Let $f:\N\rightarrow \R_+$ be a positive-valued function defined over natural numbers. Then, the optimum value of the linear program (\ref{LP-lemma}) defined below is  $\min_{2 \leq j \leq r } (1-\frac{j}{\gamma-r+j})f(n-j+1)$.
	
    \begin{alignat}{3}\label{LP-lemma}
	    & \underset{x_2, \dots, x_r, y_2, \dots, y_r}{\text{minimize}} \quad
		 && \sum_{j=2}^r (x_j-y_j)f(n-j+1)  \tag{$LP_1$} \\
		& \text{subject to} \quad
		 && 0 \leq y_j \leq x_j \ \forall j\in \set{2,\ldots, r} \label{con:0_leq_y_leq_x}\\
		 & \quad && \sum_{j=2}^r x_j = 1 \label{con:xs_eq_1} \\
		& \quad && \gamma \sum_{j=2}^r y_j \leq \sum_{j=2}^r j \cdot x_j \label{con:tys_leq_xs}
	\end{alignat}
	\end{restatable}
	
		\section{Multiobjective Min-Cuts and Pareto-Optimal Cuts}\label{sec:edge-costs}
	In this section, we give upper and lower bounds on the number of multiobjective min-cuts and pareto-optimal cuts and prove Theorem \ref{thm:num_multi}. 
	
	Let $G = (V,E)$ be an $r$-rank hypergraph and let $c_1, \dots, c_t: E \rightarrow \mathbb{R}_+$ be cost functions on the hyperedges of $G$. 
	For a subset $F$ of hyperedges, we will use $c_i(F)$ to denote $\sum_{e \in F} c_i(e)$. For a subset $U$ of vertices, we will use $\complement{U}$ to denote $V \setminus U$ and $\delta(U)$ to denote the set of hyperedges that intersect both $U$ and $\complement{U}$. For a vertex $v$, we will use $\delta(v)$ to denote $\delta(\{v\})$. A subset $F$ of hyperedges is a cut if there exists a partition $(U,\complement{U})$ such that $F = \delta(U)$. 
	We refer the reader to Section \ref{sec:intro} for the definitions of $b$-multiobjective min-cuts, multiobjective min-cuts, and pareto-optimal cuts. 
	
	We begin with a randomized algorithm for the $b$-multiobjective min-cut problem in Section \ref{sec:bmultiobjective}. We take an alternative viewpoint of this randomized algorithm in Section \ref{sec:genmultiobjective} to prove Theorem \ref{thm:num_multi}.
	Since all pareto-optimal cuts are multiobjective min-cuts, Theorem \ref{thm:num_multi} also implies that the number of pareto-optimal cuts in an $r$-rank $n$-vertex hypergraph $G$ with $t$ hyperedge-cost functions is $O(r2^{rt}n^{3t-1})$. In Section \ref{sec:verifying}, we give randomized polynomial-time algorithms to enumerate all pareto-optimal cuts and all multiobjective min-cuts. In Section \ref{sec:LowerBounds}, we show a lower bound of $\Omega(n^t)$ on the number of pareto-optimal cuts and on the number of $b$-multiobjective min-cuts. 
	
	\subsection{Finding \texorpdfstring{$b$-Multiobjective Min-Cuts}{b-Multiobjective Min-Cuts}} \label{sec:bmultiobjective}
	In this section, we design a randomized algorithm for the $b$-multiobjective min-cut problem, which is formally defined below. We use Algorithm \ref{fig:alg_b_multi}. We summarize its  correctness and run-time guarantees in Theorem \ref{thm:suc_prob}.
	
	\begin{problem}{$b$-Multiobjective Min-Cut}
	Given: A hypergraph $G=(V,E)$ with hyperedge-cost functions $c_1, \ldots, c_t:E\rightarrow \R_+$ and a budget-vector $b\in \R^{t-1}_+$.
	
	Goal: A $b$-multiobjective min-cut. 
	\end{problem}
	
\begin{figure}[ht] 
\centering\small
\begin{algorithm} 
\textul{\textsc{$b$-Multiobjective-Min-Cut}($G, r, t, c, b$):}\+
\\  {\bf Input: } An $r$-rank hypergraph $G=(V,E)$, hyperedge-cost functions \+ \+ 
\\ $c_1, \ldots, c_{t}: E\rightarrow \R_+$ and a budget-vector $b \in \R_+^{t-1}$. \- \-
\\ \rule{0pt}{3ex}If $|V| \leq rt$:\+
\\      $X \gets $ a random subset of $V$ \\      return $\delta(X)$\-
\\ For $i=1, \dots, t-1$:\+
\\ $U_i \gets \set{v \in V \setminus (\bigcup_{j=1}^{i-1} U_j) \colon c_i(\delta(v)) > b_i \}}$ \-
\\ $U_t \gets V \setminus \bigcup_{j=1}^{t-1} U_j$
\\ $i \gets \arg\max_{j\in [t]} |U_j|$
\\ $e \gets$ a random hyperedge chosen according to $\Pr[e = e'] = \frac{c_i(e')}{c_i(E)}$
\\ $G' \gets G / e$
\\ Return \textsc{$b$-Multiobjective-Min-Cut}($G', r, t, c, b$) \-
\end{algorithm}
\caption{\textsc{$b$-Multiobjective Min-Cut}}
\label{fig:alg_b_multi}
\end{figure}
	
	\begin{theorem}\label{thm:suc_prob}
	    Let $G=(V,E)$ be an $r$-rank $n$-vertex hypergraph with hyperedge-cost functions $c_1, \ldots, c_t:E\rightarrow \R_+$ and let $b\in \R^{t-1}_+$ be a budget-vector. Let $F$ be an arbitrary $b$-multiobjective min-cut. Then, Algorithm \ref{fig:alg_b_multi} outputs $F$ with probability at least $Q_n$, where
		\[ 
		Q_n := \begin{cases}
		\frac{1}{2^{rt}} &\text{if } n \leq rt, and  \\
		\frac{2t+1}{2^{rt}(rt+1)}{n-t(r-2) \choose 2t}^{-1} &\text{if } n > rt.
		\end{cases}
		\]
		Moreover, the algorithm can be implemented to run in polynomial time. 
	\end{theorem}
	
	\begin{proof}
    	We note that sets $U_i$ can be computed in polynomial time, and the algorithm recomputes them at most $n$ times. Random contraction can also be implemented in polynomial time, and therefore the overall run-time of the algorithm is polynomial.
    	
	    We now bound the correctness probability by induction on $n$. 
		For the base case, we consider $n \le rt$. In this case, the algorithm returns $\delta(X)$ for a random $X \subseteq V$. There are $2^n$ possible values for $X$, and $F = \delta(X)$ for at least one of them. Thus, the probability that the algorithm returns $F$ is at least $\frac{1}{2^n} \geq \frac{1}{2^{rt}} = Q_n$.
		
		Next, we prove the induction step. Suppose $n > rt$. We will need the following claim.
	    	\begin{claim}\label{claim:par_alg_lp}
	The algorithm outputs $F$ with probability at least the optimum value of the following linear program.
	\begin{alignat*}{3} \label{LP-multiobjective}
	    & \underset{x_2, \dots, x_r, y_2, \dots, y_r}{\text{minimize}} \quad
		 && \sum_{j=2}^r (x_j-y_j)Q_{n-j+1}  \tag{$LP_2$} \\
		& \text{subject to} \quad
		 && 0 \leq y_j \leq x_j \ \forall j\in \set{2, \ldots, r} \\
		& \quad && \sum_{j=2}^r x_j = 1 \\
		& \quad && |U_i|\sum_{j=2}^r y_j \leq \sum_{j=2}^r j \cdot x_j 
	\end{alignat*}
	
	\end{claim}
	
	\begin{proof}
	Since $n>rt$, when the algorithm is executed on $G$ it will contract a randomly chosen hyperedge and recurse. Let $e'$ be the random  hyperedge chosen by the algorithm. If $e' \not\in F$, then $F$ will still be a $b$-multiobjective min-cut in $G / e'$. We observe that $G/e'$ is a hypergraph with $n-|e'|+1$ vertices and the rank of $G/e'$ is at most the rank of $G$. Therefore, if $e' \not\in F$, then the algorithm will output $F$ with probability at least $Q_{n-|e'|+1}$.
	
	Let $i\in [t]$ be the index of the cost function chosen by the algorithm. Let 
		\begin{align*}
		    E_{j} &:= \{e \in E \colon |e| = j \}, \\
		    x_j &:= \Pr[e' \in E_j] = \frac{c_i( E_j)}{c_i(E)}, 		    \text{ and}\\
            y_j &:= \Pr[e' \in E_j \cap F] = \frac{c_i(E_j \cap F)}{c_i(E)}.
		\end{align*}
	We note that $E_{j}$ is the set of hyperedges of size $j$, $x_j$ is the probability of picking a hyperedge of size $j$ to contract, and $y_j$ is the probability of picking a hyperedge of size $j$ from $F$ to contract. We know that
		\begin{equation}
		\Pr[\text{Algorithm returns the cut } F] \geq \sum_{j=2}^r (x_j-y_j)Q_{n-j+1} .
		\label{eqn:lower_bound}
		\end{equation}
		
		The values of $x_j$ and $y_j$ will depend on the structure of $G$.  However we can deduce some relationships between them. Since $0 \leq c_i(E_j \cap F) \leq c_i(E_j)$ for every $j \in \{2, \dots, r\}$, we know that 
		\begin{align}\label{ineq:y-atmost-x}
		    0 \leq y_j \leq x_j \text{ for every } j\in \{2, \dots, r\}.
		\end{align}
		Moreover, $x_j$ is the probability of picking a hyperedge of size $j$. Hence, 
		\begin{align}\label{ineq:sum-x}
		    \sum_{j=2}^r x_j = 1.
		\end{align}
		
		Next, we show that for every $i\in [t]$ and every $v\in U_i$, we have 
		\begin{align}\label{ineq:F-atmost-delta-v}
		c_i(F)&\le c_i(\delta(v)).
		\end{align}
		If $i < t$, then $c_i(F) \leq b_i < c_i(\delta(v))$ for every $v \in U_i$. Let $i=t$. We recall that $F$ is a $b$-multiobjective min-cut. Since every cut induced by a single vertex in $U_t$ satisfies all of the budgets, no such cut can have a better $c_t$-cost than $F$, so again $c_i(F) \leq c_i(\delta(v))$ for every $v \in U_i$. 
		
		From inequality (\ref{ineq:F-atmost-delta-v}), 
		we conclude that
		
		\[ c_i(F) \leq \frac{\sum_{v \in U_i} c_i(\delta(v))}{|U_i|} \leq \frac{\sum_{v \in V} c_i(\delta(v))}{|U_i|} = \frac{\sum_{e \in E} |e|c_i(e)}{|U_i|} = \frac{\sum_{j=2}^r j \cdot c_i(E_j)}{|U_i|}. \]
		Therefore
		\[
		\sum_{j=2}^r y_j = \Pr[e' \in F] = \frac{c_i(F)}{c_i(E)} \leq  \frac{1}{|U_i|}\sum_{j=2}^r j \cdot x_j.
		\]
		Thus, we have that 
		\begin{align} \label{ineq:u_i-sum-y_j}
		|U_i|\sum_{j=2}^{r} y_j \leq \sum_{j=2}^r j \cdot x_j. 
		\end{align}
	
		The minimum value of our lower bound in equation (\ref{eqn:lower_bound}) over all choices of $x_j$ and $y_j$ that satisfy inequalities (\ref{ineq:y-atmost-x}), (\ref{ineq:sum-x}), and (\ref{ineq:u_i-sum-y_j}) is a lower bound on the probability that the algorithm outputs $F$.
	\end{proof}

	    Let $U_i$ be a the largest among the sets $U_1, \dots, U_t$ that the algorithm generates when executed on input $(G,r,t,c,b)$. Claim \ref{claim:par_alg_lp} tells us that the algorithm outputs $F$ with probability at least the optimum value of the linear program (\ref{LP-multiobjective}) from the claim.
		
		The linear program (\ref{LP-multiobjective}) is exactly the linear program (\ref{LP-lemma}) from Lemma \ref{lem:lp_sol} with $\gamma = |U_i|$ and $f(n):=Q_n$. To apply Lemma \ref{lem:lp_sol}, we just need to show that $n \geq |U_i| \geq r+1$. We recall that $U_i$ is the largest of the $t$ sets that the algorithm constructs. Each of these sets is a subset of $V$, so we can conclude that  $|U_i| \leq |V| = n$. We also know from the the construction of the sets $U_1, \dots, U_t$ that they partition $V$. This means that $\sum_{j=1}^t |U_j| = n$. Since $U_i$ is the largest of the sets, we must have $|U_i| \geq \frac{n}{t}$. Since $n > rt$, this means $|U_i| \geq \frac{rt+1}{t} > r$. Thus $|U_i| > r$, and since $r$ and $|U_i|$ are integers, we conclude that $|U_i| \geq r+1$. Therefore, we can apply Lemma \ref{lem:lp_sol} with $\gamma = |U_i|$ to conclude that 
		\[
		\Pr[\text{Algorithm} \text{ returns the cut } F] \geq \min_{2 \leq j \leq r} \left(1-\frac{j}{|U_i|-r+j}\right)Q_{n-j+1}.
		\]
		The following claim completes the proof of the theorem. 
		\end{proof}
		\begin{claim}
		For every $j\in \{2, \dots, r\}$, we have 
		 \[
		 \left(1-\frac{j}{|U_i|-r+j}\right)Q_{n-j+1}\ge Q_n
		 \]
		\end{claim}
		\begin{proof}
		Let $j \in \{2, \dots, r\}$. The given inequality is equivalent to $\frac{Q_{n-j+1}}{Q_n} \geq \frac{|U_i|-r+j}{|U_i|-r}$. Since $U_i$ is the largest among $U_1, \ldots, U_t$ which together partition $V$, we have $|U_i| \geq \frac{n}{t}$. Consequently, $\frac{|U_i|-r+j}{|U_i|-r} = 1 + \frac{j}{|U_i|-r} \leq 1+ \frac{jt}{n-rt}$. Therefore, it suffices to prove that $\frac{Q_{n-j+1}}{Q_n} \geq 1 + \frac{jt}{n-tr}$. We case on the value of $n-j+1$. 
		\begin{casesp}
			\item Suppose that $n-j+1 > rt$. Then, we have
			\begin{equation}\label{eqn:qs_ratio}
			\frac{Q_{n-j+1}}{Q_n} = \frac{ {n-t(r-2) \choose 2t} }{ {n-j+1-t(r-2) \choose 2t} }  = \prod_{\ell=0}^{2t-1} \frac{n-t(r-2)-\ell}{n-j+1-t(r-2)-\ell}.
			\end{equation}
			
			We consider two sub-cases based on the value of $j$.
			
			\begin{casesp}
				\item Suppose that $j > 2t$. Then, we observe that
				\begin{align*}
				\prod_{\ell=0}^{2t-1} \frac{n-t(r-2)-\ell}{n-j+1-t(r-2)-\ell} &\geq  \left( \frac{n-t(r-2)}{n-j+1-t(r-2)} \right)^{2t} \\
				&= \left( 1 + \frac{j-1}{n-j+1-t(r-2)} \right)^{2t} \\
				&\geq 1 + \frac{2t(j-1)}{n-j+1-t(r-2)} \\
				&\geq 1+ \frac{jt + (j-2)t}{n-rt} \\
				&\geq 1+ \frac{jt}{n-rt}.
				\end{align*}
				We use $j > 2t$ in the second to last inequality and $j \geq 2$ in the final inequality.
				
				\item Suppose that $j \leq 2t$. Then we can cancel additional terms from the the right hand side of equation (\ref{eqn:qs_ratio}) to obtain that
				\begin{align*}
				\prod_{\ell=0}^{2t-1} \frac{n-t(r-2)-\ell}{n-j+1-t(r-2)-\ell} &= \prod_{\ell=0}^{j-2} \frac{n-t(r-2)-\ell}{n-tr-\ell} \\
				&\geq \left( \frac{n-t(r-2)}{n-tr} \right)^{j-1} \\
				&= \left( 1 + \frac{2t}{n-rt} \right)^{j-1} \\
				&\geq 1 + \frac{2t(j-1)}{n-rt} \\
				&\geq 1 + \frac{jt}{n-rt}.
				\end{align*}
			\end{casesp}
			Thus, in either subcase, our desired inequality holds. 
			
			\item Suppose that $n-j+1 \leq rt$. 
			Now the expression for $Q_{n-j+1}$ is different. Since we still know that $n \geq rt+1$ and $j \leq r$, we conclude that
			\begin{align*}
			\frac{Q_{n-j+1}}{Q_n} &= \frac{(rt+1){n-t(r-2) \choose 2t}}{2t+1} \\
			&\geq \frac{(rt+1){rt+1-t(r-2) \choose 2t}}{2t+1} \\
			&= rt+1 \\
			&\geq 1 + jt \\
			&\geq 1 + \frac{jt}{n-rt}.
			\end{align*}
		\end{casesp}
		Thus, our desired inequality holds in all cases.
	\end{proof}
	
	\subsection{Finding Multiobjective Min-Cuts} \label{sec:genmultiobjective}
	In this section, we present Algorithm \ref{fig:alg_multi}, which does \emph{not} take a budget-vector as input, yet outputs any multiobjective min-cut (for any choice of budget-vector) with inverse polynomial probability. This is accomplished by returning a collection of cuts. 

	In contrast to Algorithm \ref{fig:alg_b_multi}, all of the randomness in Algorithm \ref{fig:alg_multi} (except for the selection of a random cut in the base case) occurs upfront through the selection of a permutation of the hyperedges.
	Theorem \ref{thm:pareto_alg} summarizes the guarantees of Algorithm \ref{fig:alg_multi}. 
	
\begin{figure*}[ht] 
    \centering\small
    \begin{algorithm} 
    \textul{\textsc{Multiobjective-Min-Cut}($G, r, t, c_1, \dots, c_t$):}\+
    \\  {\bf Input: } An $r$-rank hypergraph $G=(V,E)$ and \+ \+
    \\ hyperedge-cost functions $c_1,\ldots, c_t:E\rightarrow \R_+$. \- \-
    \\ \rule{0pt}{3ex}If $|V| \leq rt$: \+
    \\ Pick a random subset $X$ of $V$ and return $\delta(X)$\-
    \\ For $i=1, \dots, t$: \+
    \\ $E_i \gets \{e \in E \colon c_i(e) > 0 \}$
    \\ $\pi_i \gets$ a permutation of $E_i$ generated by repeatedly choosing a not yet \+ \\ chosen hyperedge $e$ with probability proportional to $c_i(e)$ \-\-
    \\ $R \gets \emptyset$
    \\ For each sequence $n_1, \dots n_t$ with $n \geq n_1 \geq n_2 \geq \dots \geq n_{t-1} \geq n_t = rt$: \+
    \\ $G' \gets G$ 
    \\ For $i = 1, \dots, t$: \+
    \\ While $|V(G')| > n_i$: \+
    \\ $e \gets $ the first hyperedge from $\pi_i$ that is still present in $G'$
    \\ $G' \gets G' / e$ \-\-
    \\ $X \gets $ a random subset of $V(G')$
    \\ Add $\delta(X)$ to $R$ if it is not already present \-
    \\ Return $R$ \-
    \end{algorithm}
    \caption{\textsc{Multiobjective Min-Cut}}
    \label{fig:alg_multi}
\end{figure*}
	
	\begin{theorem}\label{thm:pareto_alg}
	    Let $G$ be an $r$-rank, $n$-vertex hypergraph with $t$ hyperedge-cost functions. 
	    Then, a fixed multiobjective min-cut $F$ is in the collection returned by Algorithm \ref{fig:alg_multi} with probability 
	    \[
	    \frac{\Omega(n^{-2t})}{r2^{rt}}.
	    \]
	    Moreover, the algorithm outputs at most $n^{t-1}$ cuts. 
	\end{theorem}
	
	\begin{proof}
		We begin by showing the second part of the theorem. The algorithm outputs at most one cut for each choice of $n_1, \ldots, n_{t-1}\in [n]$ (or just one cut if $|V| \leq rt$). Hence, it outputs at most $n^{t-1}$ cuts.
		
		We now argue that the algorithm retains the same success probability as Algorithm \ref{fig:alg_b_multi}, for any fixed budget-vector $b$. 
		
		Suppose $n \leq rt$. Then both Algorithm \ref{fig:alg_multi} and Algorithm \ref{fig:alg_b_multi} return $\delta(X)$ for a random subset $X$ of the vertices of $G$. Thus, for any cut $F$, the two algorithms have the same probability of returning $F$. Henceforth, we assume $n > rt$.
		
		We will view Algorithm \ref{fig:alg_b_multi} from a different perspective. In that algorithm, whenever we contract a hyperedge $e$, we choose, for some $i\in [t]$, a hyperedge according to the probability distribution $\Pr[e = e'] = \frac{c_i(e')}{c_i(E)}$. In particular, the choice of $i$ depends on which contractions have been made so far, but the choice of a particular hyperedge, given the choice of $i$, does not depend on our previous contractions, except for the fact that we do not contract hyperedges which have already been reduced to singletons. We note that allowing the contraction of singletons would not change the success probability of the algorithm. 
		
		Therefore, we could modify Algorithm \ref{fig:alg_b_multi} so that it begins by selecting permutations $\pi_1, \dots, \pi_t$ of $E_1, \dots, E_t$ (where $E_i = \{ e \in E \colon c_i(e) > 0 \}$) as in Algorithm \ref{fig:alg_multi}, and then whenever the algorithm asks to
		contract a random hyperedge with probability proportional to its  weight under $c_i$, we  instead contract the next hyperedge from $\pi_i$ which is still present in the current hypergraph. This modification does not change at any step the probability that a particular hyperedge is the next contraction of a non-singleton hyperedge, and therefore the success probability of the algorithm remains exactly the same.
		
		Viewing Algorithm \ref{fig:alg_b_multi} in this way, we note that when we reach the base case of $n \leq rt$, we will have contracted the first $m_i$ hyperedges of each $\pi_i$, for some $m_1, \dots, m_t \in \{0, \dots |E|\}$. The crucial observation now is that interleaved contractions can be separated. That is, if we know $m_i$ for every $i\in [t]$, the order in which we do the contractions is irrelevant: we will get the same resulting hypergraph if we contract the first $m_1$ hyperedges from $\pi_1$, then contract the first $m_2$ hyperedges from $\pi_2$, and so on up through the first $m_t$ hyperedges from $\pi_t$ instead of the interleaved contractions. Let $n_1$ be the number of vertices in the hypergraph obtained after contracting the first $m_1$ hyperedges from $\pi_1$, subsequently, let $n_2$ be the number of vertices in the hypergraph obtained after contracting the first $m_2$ hyperedges from $\pi_2$ and so on. 
		
		When we view Algorithm \ref{fig:alg_b_multi} in this way, it is only the choice of the values $n_1, \dots, n_t$ that depends on the budgets, while the choice of the permutation $\pi_i$ does not depend on the budgets. Algorithm \ref{fig:alg_multi} is running exactly the version of Algorithm \ref{fig:alg_b_multi} that we have just described, except that instead of choosing $n_1, \dots, n_t$ based on the budgets, it simply tries all possible options (which will certainly include whichever $n_1, \dots, n_t$ Algorithm \ref{fig:alg_b_multi} would use for the given input budget-vector). 
		
		Therefore for every fixed choice of budget-vector $b$, and every fixed $b$-multiobjective min-cut $F$, the probability that $F$ is in the collection $R$ output by the algorithm is at least as large as the probability that $F$ is the cut output by Algorithm \ref{fig:alg_b_multi}. By Theorem \ref{thm:suc_prob}, this probability is $\frac{\Omega(n^{-2t})}{r2^{rt}}$, as desired.
	\end{proof}
	
	We derive Theorem \ref{thm:num_multi} from Theorem \ref{thm:pareto_alg} now. 
	
	\numMultiobjective*
	\begin{proof}
	Let $x$ be the number of multiobjective min-cuts in $G$. By Theorem \ref{thm:pareto_alg}, the expected number of multiobjective min-cuts output by our algorithm  \textsc{Multiobjective Min-Cut} (i.e., Algorithm \ref{fig:alg_multi}) is $(x/r2^{rt})\Omega(n^{-2t})$. Theorem \ref{thm:pareto_alg} also tells us that the algorithm outputs at most $n^{t-1}$ cuts. Therefore, $x=r2^{rt} \cdot O(n^{3t-1})$.
	\end{proof}
	
	\subsection{Enumerating Multiobjective Min-Cuts and Pareto-Optimal Cuts}\label{sec:verifying}
	
	In this section, we give algorithms to enumerate all multiobjective min-cuts and pareto-optimal cuts in polynomial time.
	
	We first give a polynomial time algorithm to enumerate all multiobjective min-cuts. We execute our algorithm for \textsc{Multiobjective Min-Cut} (i.e., Algorithm \ref{fig:alg_multi}) a sufficiently large number of times so that it succeeds with high probability (i.e., with probability at least $1-1/n$): In particular, executing it $r^2t2^{rt}O(n^{2t}\log n)$ 
	many times gives us a collection $\mathcal{C}$ that is a superset of the collection $\mathcal{C}_{MO}$ of multiobjective min-cuts with high probability.  Moreover, the size of the collection $\mathcal{C}$ is $r^2t2^{rt}O(n^{3t-1}\log n)$. We can prune $\mathcal{C}$ to identify $\mathcal{C}_{MO}$ in polynomial-time as follows: remove every cut $F \in \mathcal{C}$ for which there exists a cut $F'\in \mathcal{C}$ with $c_t(F') < c_t(F)$ and $c_i(F') \leq c_i(F)$ for every $1 \leq i \leq t-1$. 
	
	Next, we give a polynomial time algorithm to enumerate pareto-optimal cuts. By containment relation (\ref{containment-relationships}), it suffices to identify all pareto-optimal cuts in the collection $\mathcal{C}$. For this, we only need a polynomial-time procedure to verify if a given cut $F$ is pareto-optimal. Algorithm \ref{fig:alg_verify_pareto_cut} gives such a procedure. It essentially searches for a cut that dominates the given cut $F$ by running our algorithm for \textsc{$b$-Multiobjective Min-Cut} with $t$ different budget-vectors. 
	
	\begin{figure*}[ht] 
    \centering\small
    \begin{algorithm} 
    \textul{\textsc{Verify-Pareto-Optimality}($G, r, t, c_1, \dots, c_t, F$):}\+
    \\  {\bf Input: } An $r$-rank hypergraph $G=(V,E)$, cost functions $c_1,\ldots, c_t:E\rightarrow \R_+$, \+ \+ \\ and a cut $F$ in $G$ \- \-
    \\ \rule{0pt}{3ex}For $i=1, \dots, t$: \+
    \\ $\vec{c} \gets (c_1, \dots c_{i-1}, c_{i+1}, \dots, c_t, c_i)$
    \\ $\vec{b} \gets (c_1(F), \dots, c_{i-1}(F), c_{i+1}(F), \dots, c_t(F))$
    \\ For $j=1, \dots, r2^{rt}O(n^{2t}\log(n))$: \+
    \\ $F' \gets $ \textsc{$b$-Multiobjective-Min-Cut}($G,r,t, \vec{c}, \vec{b}$)
    \\ If $F'$ is a $b$-multiobjective cut in $(G,\vec{c})$ and $c_i(F') < c_i(F)$: \+
    \\ Return FALSE \-\-\-
    \\ Return TRUE \-
    \end{algorithm}
    \caption{Verify pareto-optimality of a given cut}
    \label{fig:alg_verify_pareto_cut}
\end{figure*}
	
	\begin{theorem}\label{thm:verify_cut}
		Given an $r$-rank, $n$-vertex hypergraph $G$ with $t$ hyperedge-cost functions and a cut $F$ in $G$, if $F$ is a pareto-optimal cut, then Algorithm \ref{fig:alg_verify_pareto_cut} returns TRUE, and if $F$ is not a pareto-optimal cut, then the algorithm returns FALSE with high probability. Moreover, the algorithm can be implemented to run in polynomial time.
		
	\end{theorem}
	
	\begin{proof}
	    The run-time of the algorithm is polynomial since our algorithm for \textsc{$b$-Multiobjective Min-Cut} (i.e., Algorithm \ref{fig:alg_b_multi}) is a polynomial-time algorithm. Algorithm \ref{fig:alg_verify_pareto_cut} returns false only if it finds a cut that dominates the input cut $F$. If $F$ is pareto-optimal, no such cut will exist, and therefore the algorithm will return true. 
		
		Next, suppose the input cut $F$ is not pareto-optimal. Then $F$ must be dominated by some other cut $F'$. Let $i \in [t]$ be such that $c_i(F') < c_i(F)$ (such an $i$ is guaranteed to exist by the definition of domination). Let $b$ be a budget-vector of the costs of $F$ under the cost functions other than $c_i$ and let $c'$ be $c$ with $c_i$ moved to the end of the vector of cost functions. Then $F$ will not be a $b$-multiobjective min-cut in $(G,c')$, since $F'$ also satisfies $b$, but has a lower $c_i$-cost. 
		
		Therefore, any $b$-multiobjective min-cut will dominate $F$ (since it will also satisfy $b$ and cannot have higher $c_i$ cost than $F'$). Thus, if any of our $r2^{rt}n^{2t}\log(n)$ calls to our algorithm for \textsc{$b$-Multiobjective-Min-Cut} (i.e., Algorithm \ref{fig:alg_b_multi}) for this value of $i$ returns a multiobjective min-cut, then the algorithm will return false. By Theorem \ref{thm:suc_prob}, our algorithm for \textsc{$b$-Multiobjective Min-Cut} returns a $b$-multiobjective min-cut with probability $\frac{1}{r2^{rt}}\Omega(\frac{1}{n^{2t}})$. Therefore, if we run this algorithm $r2^{rt}O(n^{2t}\log(n))$ times, a $b$-multiobjective min-cut will be returned at least once with high probability, and our algorithm will correctly return false. 
	\end{proof}
	
	\subsection{Lower Bounds}\label{sec:LowerBounds}
	
	In this section we discuss lower bounds on the number of distinct pareto-optimal cuts in $n$-vertex hypergraphs. Karger gave a family of graphs with $n^{t/2}$ parametric min-cuts \cite{Karger2016minimish}. 
	We recall that every parametric min-cut is a pareto-optimal cut by the containment relation (\ref{containment-relationships}). Thus, $n^{t/2}$ is also a lower bound on the number of pareto-optimal cuts in $n$-vertex hypergraphs. To the best of the authors' knowledge, this is the best lower bound on the number of pareto-optimal cuts that is known in the literature. We give an $\Omega(n^t)$ (for constant $t$) lower bound on the number of pareto-optimal cuts in a graph. Our lower bound construction is different from that of Karger. 
	
	\begin{theorem}\label{thm:pareto_lower_bound}
	For all positive integers $t$ and $n$ such that $n \geq t+2$, there exists an $n$-vertex graph $G$ with associated edge-cost functions $c_1, \dots, c_t : E(G) \rightarrow \mathbb{R}_+$ such that $G$ has at least $\left( \frac{n-2}{t} \right)^t$ distinct pareto-optimal cuts.
	\end{theorem}

	\begin{proof}
	For fixed $n$ and $t$, construct a graph $G$ as follows. The graph $G$ has two special vertices $u$ and $v$. The rest of the vertices are used to form $t$ distinct paths between $u$ and $v$ with each path consisting of at least $\lfloor \frac{n-2}{t} \rfloor +1 > \frac{n-2}{t}$ distinct edges. 
	We assign edge costs as follows: If $e$ is an edge in the $i$'th path, then $c_i(e) = 1$, while $c_j(e) = 1/(t+1)$ for every $j\in [t]\setminus \set{i}$. See Figure \ref{fig:lower_bound} for an example.
	\renewcommand{\figurename}{Figure}
	\begin{figure}
	    \centering
	    \includegraphics[trim={0 2cm 4cm 2cm}, clip, scale=0.4]{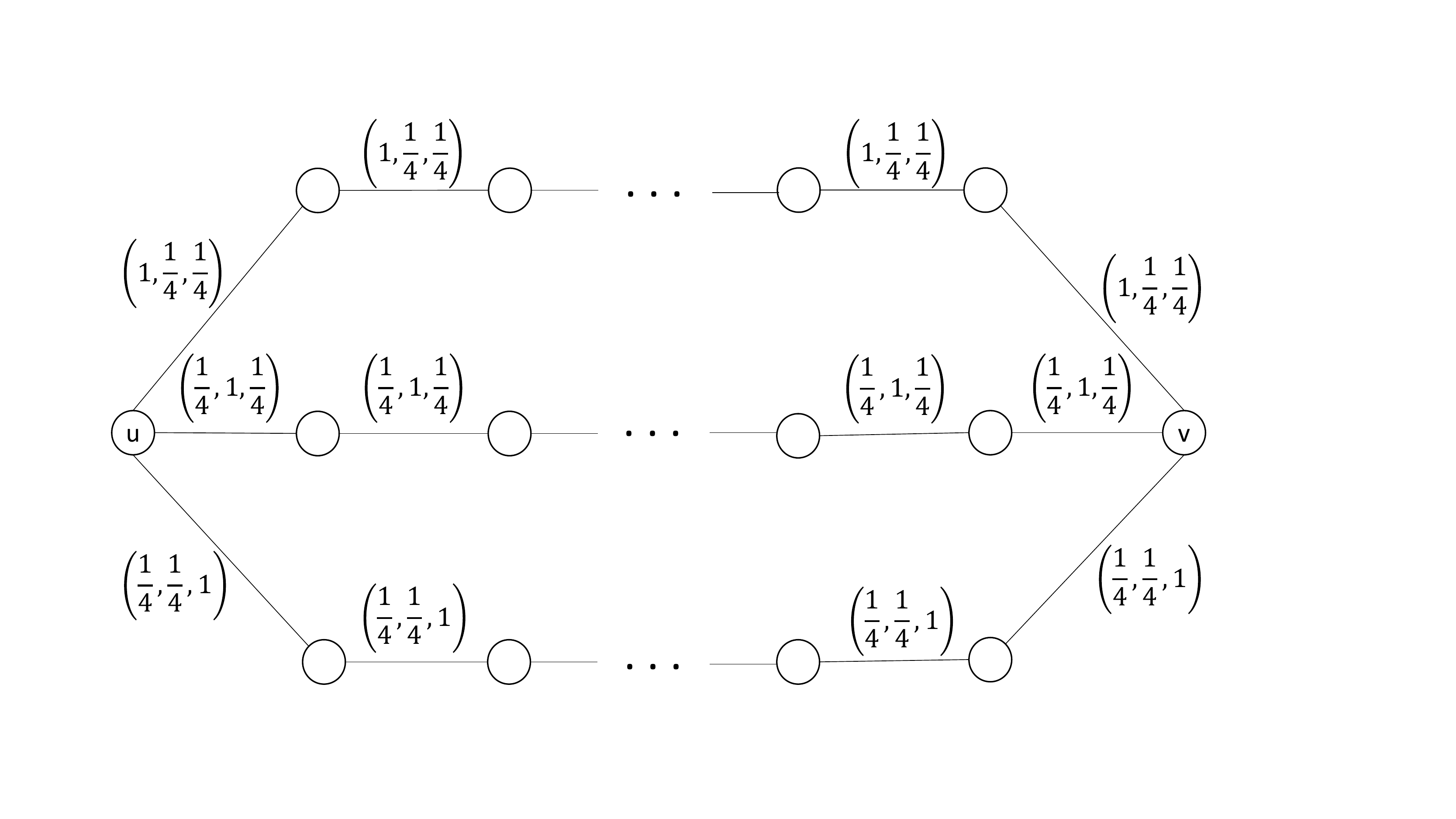}
	    \caption{An illustration of our lower bound construction for $t=3$.}
	    \label{fig:lower_bound}
	\end{figure}
	
    \renewcommand{\figurename}{Algorithm}
	
	We will show that any cut which contains exactly one edge from each path is pareto-optimal. The number of such cuts is at least $\left( \frac{n-2}{t} \right)^t$, since each path has at least $(n-2)/t$ edges, so this will suffice to prove the theorem.
	
	We observe that any cut contains either exactly one edge from each path or at least two edges from some path.
	Any cut $F$ which contains exactly one edge from each path will have $c_i(F)=2t/(t+1)$ for every $i \in [t]$. 
    Any cut $F'$ that contains at least two edges from some path $i\in [t]$ will have  $c_i(F')=2>2t/(t+1)$. Therefore no cut which contains two edges from the same path can dominate a cut which contains exactly one edge from each path. Furthermore, if two different cuts each contain exactly one edge from all paths, then they both have the same cost under every cost function, and thus neither can dominate the other. We conclude that every cut which contains exactly one edge from each path is pareto-optimal.
    
    \end{proof}
	
	\noindent \textbf{Remark 1.} The lower bound from Theorem \ref{thm:pareto_lower_bound} is still significantly smaller than the $O(n^{3t-1})$ upper bound from Theorem \ref{thm:num_multi}. We believe that this gap comes from the slack in the analysis of our randomized algorithms.\\
	
	\noindent \textbf{Remark 2.} We note that the construction in Theorem \ref{thm:pareto_lower_bound} also shows that there exists a budget-vector $b\in \R^{t-1}_+$ such that the number of $b$-multiobjective min-cuts is $\Omega(n^t)$: consider budget values $b_i=(2t)/(t+1)$ for every $i\in [t-1]$.  
	We emphasize that since not every multiobjective min-cut is pareto-optimal, this lower bound does not imply the one from Theorem \ref{thm:pareto_lower_bound}. Since distinct pareto-optimal cuts need not be $b$-multiobjective min-cuts for the same vector $b$, the bound in Theorem \ref{thm:pareto_lower_bound} does not immediately imply this bound either.

		\section{Node-Budgeted Multiobjective Min-Cuts}\label{sec:node_budgets}
	In this section, we give algorithms to find min-cuts that satisfy node-weighted budget constraints. Theorem \ref{thm:num_node_budgeted_multi} will be a consequence of these algorithms. 
	We begin by formally defining the problem.
	
	Let $G = (V,E)$ be a hypergraph with hyperedge-cost function $c : E \rightarrow \mathbb{R}_+$. Let $w_1, \dots, w_t : V \rightarrow \mathbb{R}_+$ be vertex-weight functions. Let $c(F) = \sum_{e \in F} c(e)$ for $F \subseteq E$, and $w_i(U) = \sum_{v \in U} w_i(v)$ for $U \subseteq V$. The following definition will be useful in defining node-budgeted multiobjective min-cuts.
	
	\begin{definition}\label{defn:feasible-vertices}
	For a budget-vector $b \in \mathbb{R}^t_+$,
	\begin{enumerate}
	    \item a vertex $v \in V$ is \emph{feasible} if $w_i(v) \leq b_i$ for all $i \in [t]$ and \emph{infeasible} otherwise and
	    \item a set of vertices $S \subseteq V$ is \emph{feasible} if $w_i(S) = \sum_{v \in S} w_i(v) \leq b_i$ for all $i \in [t]$ and \emph{infeasible} otherwise.
	\end{enumerate}
	\end{definition}
	
	We recall the definition of node-budgeted multiobjective min-cuts. 
    \begin{definition}
	    For a budget-vector $b \in \mathbb{R}_+^{t}$, a set $F \subseteq E$ is a \emph{node-budgeted} $b$-\emph{multiobjective min-cut} if $F = \delta(X)$ for some feasible set $\emptyset \subsetneq X \subsetneq V$, and $c(F)$ is minimum among all such subsets of $E$. A set $F \subseteq E$ is a \emph{node-budgeted multiobjective min-cut} if there exists a budget-vector $b\in \R^t_+$ for which $F$ is a node-budgeted $b$-multiobjective min-cut.
	\end{definition}
	
	The following will be a central problem of interest in this section.
	\begin{problem}{Node-Budgeted $b$-Multiobjective Min-Cut}
	Given: A hypergraph $G=(V,E)$ with vertex-weight functions $w_1, \dots, w_t \colon V \to \mathbb{R}_+$, a hyperedge-cost function $c \colon E\rightarrow \R_+$, and a budget-vector $b\in \R^t_+$.
	
	Goal: A node-budgeted $b$-multiobjective min-cut. 
	\end{problem}
	
	\subsection{Constant-Rank Hypergraphs}
	In this section, we design a polynomial-time algorithm to find node-budgeted $b$-multiobjective min-cuts in constant-rank hypergraphs and then prove the first part of Theorem \ref{thm:num_node_budgeted_multi}. We use Algorithm \ref{fig:alg_node_budget_b_multi_constant} to solve node-budgeted $b$-multiobjective min-cuts in constant-rank hypergraphs. It essentially runs the standard random contraction algorithm for min-cut with an additional step that deterministically contracts all infeasible vertices together. We summarize the guarantees of this algorithm in Theorem \ref{thm:node_suc_prob}. 
	We will subsequently use Theorem \ref{thm:node_suc_prob} to prove the first part of Theorem \ref{thm:num_node_budgeted_multi}.
	
	\begin{figure*}[ht] 
    \centering\small
    \begin{algorithm} 
    \textul{\textsc{Node-Budgeted-$b$-Multiobjective-Min-Cut-Constant-Rank}($G, r, t, w, c, b$):}\+
    \\  {\bf Input: } An $r$-rank hypergraph $G=(V,E)$, a positive integer $t$, \+\+ \\ a vector $w$ of vertex-weight functions with $w_i:V\rightarrow \R_+$ for $i \in [t]$, \\ a cost function $c: E \rightarrow \R_+$, and budget-vector $b \in \mathbb{R}_+^{t}$ \-\-
    \\ \rule{0pt}{3ex}Contract all infeasible vertices of $G$ into a single vertex
    \\ If $|V| \leq r+1$: \+
    \\ $X \gets$ a random subset of $V$
    \\ Return $\set{\delta(X)}$ \-
    \\ $e \gets$ a random hyperedge chosen according to $\Pr[e = e'] = \frac{c(e')}{c(E)}$ 
    \\ $(G,w) \gets (G,w) / e$
    \\ Return \textsc{Node-Budgeted-$b$-Multiobjective-Min-Cut-Constant-Rank}($G, r, t, w, c, b$)
    \end{algorithm}
    \caption{\textsc{Node-Budgeted $b$-Multiobjective Min-Cut} in constant-rank hypergraphs}
    \label{fig:alg_node_budget_b_multi_constant}
\end{figure*}
	
	\begin{theorem}\label{thm:node_suc_prob}
	Let $G=(V,E)$ be an $r$-rank $n$-vertex hypergraph with hyperedge-cost function $c:E\rightarrow \R_+$ and vertex-weight functions $w_1, \ldots, w_t: V\rightarrow \R_+$ and a budget-vector $b\in \R^{t}_+$. Let $F$ be an arbitrary node-budgeted $b$-multiobjective min-cut in $G$. Then Algorithm \ref{fig:alg_node_budget_b_multi_constant} 
	 returns $F$ with probability at least $\frac{1}{2^{r+1} {n \choose 2}}$. Moreover, the algorithm can be implemented to run in polynomial time.
	\end{theorem}
	
	\begin{proof}
	    We first analyze the run-time. Each recursive call reduces the number of vertices, so the total number of recursive calls is at most $n$. Apart from the recursion, the algorithm only performs contractions and returns a random cut, all of which can be done in polynomial time.
	
		Now we analyze the success probability. Let $Q_n := \frac{1}{2^{r+1}}{n \choose 2}^{-1}$. We will show that the algorithm returns $F$ with probability at least $Q_n$. We prove this by induction on $n$. Let $\emptyset \subsetneq X \subsetneq V$ be a feasible set with $\delta(X) = F$. We first note that all vertices in $X$ must be feasible. Therefore, the cut $X$ cannot be destroyed when all infeasible vertices are contracted together. This means that if $G$ has multiple infeasible vertices, they will simply be contracted to yield a smaller hypergraph with at most one infeasible vertex where $F$ is still a node-budgeted $b$-multiobjective min-cut. 
		Therefore, we will assume without loss of generality that $G$ contains at most one infeasible vertex.
		
		For the base case, we consider $n \le r+1$. In this case, the algorithm returns $\delta(X)$ for a random $X \subseteq V$. There are $2^n$ possible choices for $X$, and $F$ for at least one of them. Thus, the probability that the algorithm returns $F$ is at least  $\frac{1}{2^n} \geq \frac{1}{2^{r+1}} \geq Q_n$.
		
		We now prove the induction step. Let $n>r+1$ and assume that the theorem holds for all hypergraphs with at most $n-1$ vertices and rank at most $r$. We begin by showing the following claim.
		
		\begin{claim}\label{claim:node_alg_lp}
	The algorithm outputs $F$ with probability at least the optimum value of the following linear program.
	\begin{alignat*}{3} \label{nodeLP}
	    & \underset{x_2, \dots, x_r, y_2, \dots, y_r}{\text{minimize}} \quad
		 && \sum_{j=2}^r (x_j-y_j)Q_{n-j+1}  \tag{$LP_3$} \\
		& \text{subject to} \quad
		 && 0 \leq y_j \leq x_j \ \forall j\in \set{2, \ldots, r} \\
		& \quad && \sum_{j=2}^r x_j = 1 \\
		& \quad && (n-1)\sum_{j=2}^r y_j \leq \sum_{j=2}^r j \cdot x_j 
	\end{alignat*}
	
	\end{claim}
		
		\begin{proof}
		Since $n > r+1$, the algorithm will contract a randomly chosen hyperedge and recurse. Let $e'$ be the random hyperedge chosen by the algorithm. If $e' \not\in F$, then $F$ will still be a node-budgeted $b$-multiobjective min-cut in $G/e'$. We observe that $G/e'$ is a hypergraph with $n-|e'|+1$ vertices and that the rank of $G/e
		$ is at most the rank of $G$. Therefore, if $e' \not\in F$, the algorithm will output $F$ with probability at least $Q_{n-|e'|+1}$
		
		Let 
		\begin{align*}
		    E_j :=& \{e \in E \colon |e| = j \}, \\ 
		    x_j :=& \Pr[e' \in E_j] = \frac{c(E_j)}{c(E)}, \text{ and} \\ y_j :=& \Pr[e' \in E_j \cap F] = \frac{c(E_j \cap F)}{c(E)}. 
		\end{align*}
		
		We note that $E_{j}$ is the set of hyperedges of size $j$, $x_j$ is the probability of picking a hyperedge of size $j$ to contract, and $y_j$ is the probability of picking a hyperedge of size $j$ from $F$ to contract. We know that
		
		\begin{equation} \label{eqn:node-lower-bound} \Pr[\text{Algorithm} \text{ returns the cut } F] \geq \sum_{j=2}^r (x_j-y_j)Q_{n-j+1} \end{equation}
		
		The values of $x_j$ and $y_j$ will depend on the structure of $G$.  Nevertheless, we can deduce some relationships between them. Since $0 \leq c(E_j \cap F) \leq c(E_j)$ for every $j \in \{2, \dots, r\}$, we know that 
		\begin{align}\label{ineq:node-y-atmost-x}
		    0 \leq y_j \leq x_j \text{ for every } j\in \{2, \dots, r\}.
		\end{align}
		Moreover, $x_j$ is the probability of picking a hyperedge of size $j$. Hence, 
		\begin{align}\label{ineq:node-sum-x}
		    \sum_{j=2}^r x_j = 1.
		\end{align}
		We also know that $c(F) \leq c(\delta(v))$ for every feasible vertex $v$. Since we have assumed that $G$ has at most $1$ infeasible vertex, it has at least $n-1$ feasible ones, and thus,
		
		\[c(F) \leq  \frac{\sum_{v \, : v \text{ is feasible}} c(\delta(v))}{|\{v \colon v \text{ is feasible} \}|} \leq \frac{\sum_{v \in V} c(\delta(v))}{n-1} = \frac{\sum_{e \in E} |e|c(e)}{n-1} = \frac{\sum_{j=2}^r j \cdot  c(E_j)}{n-1}. \]
		Thus we have that,
		\begin{equation} \label{eqn:node-sum-yj}
		\sum_{j=2}^r y_j = \Pr[e' \in F] = \frac{c(F)}{c(E)} \leq \frac{1}{n-1}\sum_{j=2}^r j \cdot x_j .\end{equation}
		The minimum value of our lower bound in equation (\ref{eqn:node-lower-bound}) over all choices of $x_j$ and $y_j$ that satisfy inequalities (\ref{ineq:node-y-atmost-x}), (\ref{ineq:node-sum-x}), and (\ref{eqn:node-sum-yj}) is a lower bound on the probability that the algorithm outputs $F$.
		\end{proof}
		
		Claim \ref{claim:node_alg_lp} tells us that the algorithm outputs $F$ with probability at least the optimum value of the linear program (\ref{nodeLP}) from the claim. This linear program is exactly the linear program (\ref{LP-lemma}) from Lemma \ref{lem:lp_sol} with $\gamma = n-1$ and $f(n):=Q_n$. Since $n > r+1$, we have that $n \geq n-1 \geq r+1$. Therefore, we can apply Lemma \ref{lem:lp_sol} to conclude that 
		
		\[ \Pr[\text{Algorithm} \text{ returns the cut } F] \geq \min_{j\in \{2,\ldots, r\}}\left(\left(1-\frac{j}{n-1-r+j}\right)Q_{n-j+1}\right).  \]
		
		It remains to show that $\min_{j \in \{2, \dots, r\}} ((1-\frac{j}{n-1-r+j})Q_{n-j+1}) \geq Q_n$. Let $j \in \{2, \dots, r\}$. Then it suffices to show that $\frac{Q_{n-j+1}}{Q_n} \geq \frac{n-1-r+j}{n-1-r} = 1 + \frac{j}{n-1-r}$.
		We have
		\[
		\frac{Q_{n-j+1}}{Q_n} = \frac{{n \choose 2}}{{n-j+1 \choose 2}} = \frac{n(n-1)}{(n-j+1)(n-j)} \geq \left(\frac{n}{n-j+1}\right)^2 = \left(1 + \frac{j-1}{n-j+1}\right)^2,
		\]
		and
		\[
		\left(1 + \frac{j-1}{n-j+1}\right)^2 \geq 1 + \frac{2(j-1)}{n-j+1} \geq 1 + \frac{j}{n-j+1} \geq 1 + \frac{j}{n-1-r}.
		\]
		The last two inequalities follow from the fact that $2 \leq j \leq r$. 
	\end{proof}
	
We now restate and prove the first part of Theorem \ref{thm:num_node_budgeted_multi}. At a high level, our approach will be similar to the proof of Theorem \ref{thm:num_multi}. We will modify the algorithm for \textsc{Node-Budgeted $b$-Multiobjective Min-Cut} to obtain Algorithm \ref{fig:alg_node_budget_multi_constant} which outputs a collection $\mathcal{C}$ of $r \cdot O(n^t)$-many cuts such that every node-budgeted multiobjective min-cut is in $\mathcal{C}$ with probability $\frac{1}{2^r} \cdot \Omega(\frac{1}{n^2})$.
The analysis has a few subtleties that distinguish it from the edge-budgeted version, so we include the full details.

\begin{figure*}[ht] 
    \centering\small
    \begin{algorithm} 
    \textul{\textsc{Node-Budgeted-Multiobjective-Min-Cut-Constant-Rank}($G, r, t, w, c)$):}\+
    \\  {\bf Input: } An $r$-rank hypergraph $G=(V,E)$, a positive integer $t$, \+\+ \\ a vector $w$ of vertex-weight functions with $w_i:V\rightarrow \R_+$ for $i \in [t]$, \\ and a cost function $c: E \rightarrow \R_+$ \-\-
    \\ \rule{0pt}{3ex}
    $\pi \gets$ a permutation of $E$ generated by repeatedly choosing a not yet \+ 
    \\ chosen hyperedge with probability proportional to $c(e)$ \-
    \\ $R \gets \emptyset$
    \\ $T \gets \emptyset$
    \\ For $n' = 2, \dots, n$: \+
    \\ $G' \gets G$
    \\ Contract hyperedges from $G'$ in the order given by $\pi$ until $G'$ has at most $n'$ vertices
    \\ For each $x_1, \dots, x_{t}$ such that $x_i = c_i(v)$ for some $v \in G'$ for all $i\in [t]$: \+
    \\ $G'' \gets G'$
    \\ Contract together all vertices $v$ in $G''$ which have $c_i(v) > x_i$ for some $i$
    \\ If $r+2 > |V(G'')| > 1 \text{ and } V(G'') \not\in T$: \+
    \\ $S \gets $ a random subset of $V(G')$ with $\emptyset \subset S \subset V(G')$
    \\ Add $V(G'')$ to $T$
    \\ Add $\delta(S)$ to $R$ if it is not already present \- \- \-
    \\ Return $R$ \-
    \end{algorithm}
    \caption{\textsc{Node-Budgeted Multiobjective Min-Cut} in constant-rank hypergraphs}
    \label{fig:alg_node_budget_multi_constant}
\end{figure*}
	
\begin{theorem}
The number of node-budgeted multiobjective min-cuts in an $r$-rank, $n$-vertex hypergraph with $t$ vertex-weight functions is $O(r2^r n^{t+1})$. 
\end{theorem}

\begin{proof}
		Let $G$ be an $r$-rank $n$-vertex hypergraph with vertex-weight functions $w_1, \ldots, w_t:V\rightarrow \R_+$. We will denote $w=(w_1, \ldots, w_t)$ and the hypergraph with the vertex-weight functions by the tuple $(G,w)$ for conciseness. 
		We first show that for any cut $F \subseteq E(G)$ which is a node-budgeted $b$-multiobjective min-cut in $(G,w)$ for some budget-vector $b\in \R^t_+$, the cut $F$ is among the cuts returned by Algorithm \ref{fig:alg_node_budget_multi_constant} with probability $\Omega(2^{-r}n^{-2})$.
		
		We will view Algorithm \ref{fig:alg_node_budget_b_multi_constant} from a different perspective. That algorithm alternates between contracting together infeasible vertices and contracting random hyperedges until the hypergraph has at most $r+1$ vertices. We note that the probability that a given hyperedge $e$ is the next one contracted depends only on the cost of $e$ relative to the other hyperedges. In particular it does not depend on which infeasible vertices have been contracted together. Therefore, we could modify Algorithm \ref{fig:alg_node_budget_b_multi_constant} so that it contracts random hyperedges until the hypergraph resulting from contracting all infeasible vertices together has at most $r+1$ vertices, at which point, it contracts all infeasible vertices together and returns a random cut in the resulting hypergraph. This modified version of the algorithm would retain the same success probability as the original version. In this modified algorithm, the next contraction does not depend at all on the previous contractions, so we can choose a uniform random permutation of the hyperedges at the start of the algorithm and simply contract hyperedges from that permutation until we can contract all infeasible vertices to obtain a hypergraph containing at most $r+1$ vertices.
		
		Let $U$ be the set of all feasible vertices in $G$, and for each $i \in [t]$, let $x_i = \max_{u \in U} w_i(u)$. Since all vertices in $U$ are feasible, we know that for every $u \in U$, we have $w_i(u) \leq x_i \leq b_i$ for every $i \in [t]$. Now consider an infeasible vertex $v$ in $G$. Since $v$ violates the budget-vector $b$, there must be some $i \in [t]$ such that $w_i(v) > b_i \geq x_i$. Therefore, if we wish to contract together all infeasible vertices in $G$, it suffices to find, for each $i \in [t]$, the feasible vertex $u_i$ with maximum weight under $w_i$, and then contract together all vertices whose $w_i$ weight is greater than that of $u_i$ for some $i \in [t]$. In particular, we can further modify our modified version of Algorithm \ref{fig:alg_node_budget_b_multi_constant} to use this method of contracting all infeasible vertices, and the success probability will still remain $\Omega(2^{-r}n^{-2})$. 
		
		Algorithm \ref{fig:alg_node_budget_multi_constant} is running exactly the version of the algorithm 
		that we have just described, with two additional modifications: (1) Instead of contracting hyperedges from $\pi$ until the contraction of infeasible vertices would yield a hypergraph with at most $r+1$ vertices, it simply tries all possible stopping points for the contraction of hyperedges from $\pi$, and (2) instead of choosing the values $x_1, \ldots, x_t$ based on a budget-vector $b$, it simply tries all possible values for $x_1, \ldots, x_t$. This means that, for any budget-vector $b$, Algorithm \ref{fig:alg_node_budget_multi_constant} will try the values of $n'$ and $x_1, \dots, x_t$ that the modified version of Algorithm \ref{fig:alg_node_budget_b_multi_constant} would use.
		
		Therefore, by Theorem \ref{thm:node_suc_prob}, we know that for any fixed budget-vector $b$ and every fixed node-budgeted $b$-multiobjective min-cut $F$, the probability that $F$ is among the cuts output by the algorithm is $\Omega(2^{-r}n^{-2})$.
		
		Now we bound the number of cuts returned by the algorithm. We note that the algorithm only adds a new cut to $R$ if the set of vertices that the algorithm ends up with after performing all contractions has size between $2$ and $r+1$ and is different from every set the algorithm has already obtained from previous combinations of parameters. We will show that, for a fixed $G$ and $\pi$, the number of distinct sets of  size between $2$ and $r+1$ that we can obtain by contracting vertices in the way specified by the algorithm is at most $rn^{t}$.
		
		There are at most $n$ ways to choose the value of $n'$, and also at most $n$ choices for the values of $x_1, \dots, x_{t-1}$. For fixed values of $n'$ and $x_1, \dots, x_{t-1}$, the choice of $x_t$ determines the final set of vertices after contraction. Decreasing $x_t$ can cause more vertices to become contracted (because some new vertex $v$ may now have $w_t(v) > x_t$), but it cannot cause any vertex that was previously being contracted to no longer be contracted. Thus, there are most $r$ distinct sets of vertices of size between $2$ and $r+1$ that we could obtain by varying the value of $x_t$. Therefore the total number of distinct sets of  size between $2$ and $r+1$ that could result from contracting vertices in the way described in the algorithm is at most $rn^t$.
		
		To finish the proof, let $N$ be the number of node-budgeted multiobjective min-cuts in $G$. We have shown that our algorithm outputs $\Omega(\frac{N}{2^rn^2})$ of these cuts in expectation. But since our algorithm outputs at most $rn^t$ cuts, we conclude that the number $N$ of multiobjective min-cuts must be $O(r2^rn^{t+2})$.
		
	\end{proof}

    \subsection{Arbitrary-Rank Hypergraphs}
    In this section, we present a polynomial-time algorithm for node-budgeted $b$-multiobjective min-cut in arbitrary-rank hypergraphs. 
    The second part of Theorem \ref{thm:num_node_budgeted_multi} will follow from the correctness analysis of this algorithm.

    We recall that global min-cut (without node-budgets) in arbitrary-rank hypergraphs already requires the \emph{non-uniform} random contraction technique. We extend the non-uniform contraction technique of \cite{chandrasekararen2018hypergraph} for the node-budgeted variant.
    In addition, our algorithm will use the non-uniform contraction algorithm for global min-cut by \cite{chandrasekararen2018hypergraph} as a subroutine. We reproduce their algorithm for completeness in Algorithm \ref{fig:alg_hyper_min_cut} and state its guarantee in Theorem \ref{thm:arb_rank_min_cut_edge_weights}.
	
\begin{figure*}[ht] 
    \centering\small
    \begin{algorithm} 
    \textul{\textsc{Hypergraph-Min-Cut}($G, c$):}\+
    \\  {\bf Input: } A hypergraph $G=(V,E)$, and a cost function $c: E \rightarrow \R_+$
    \\ \rule{0pt}{3ex}Compute $\beta_e := \frac{|V|-|e|}{|V|} \cdot c(e)$ for every hyperedge $e \in E$
    \\ If $\beta_e = 0$ for every hyperedge $e \in E(G)$, then return $E(G)$
    \\ $e \gets $ a random hyperedge of $G$ chosen with probability proportional to $\beta_e$
    \\ Return \textsc{Hypergraph-Min-Cut}($G/e,c$) \-
    \end{algorithm}
    \caption{\textsc{Hypergraph Min-Cut}}
    \label{fig:alg_hyper_min_cut}
\end{figure*}
	
	\begin{theorem}\label{thm:arb_rank_min_cut_edge_weights} \cite{chandrasekararen2018hypergraph}
	Algorithm \ref{fig:alg_hyper_min_cut} runs in polynomial time and returns any fixed min-cut of an $n$-vertex hypergraph $G$ with hyperedge-cost function $c$ with probability at least $\binom{n}{2}^{-1}$.
	\end{theorem}
	
	Now we describe our algorithm to solve node-budgeted $b$-multiobjective min-cut in arbitrary-rank hypergraphs. We recall that a vertex $v$ is feasible if $w_i(v)\le b_i$ for all $i\in [t]$. 
	Let $U$ be the set of all feasible vertices in $G$. We emphasize that $U$ is the set of all feasible vertices, but $U$ may not be a feasible set---see Definition \ref{defn:feasible-vertices}. Our algorithm chooses a hyperedge $e$ to contract with probability proportional to 
	\[
	\alpha_e := \left(\frac{|U|-|e \cap U|}{|U|}\right) \cdot c(e) = \left(\frac{|U \setminus e|}{|U|}\right) \cdot c(e). 
	\]
	and recurses on the contracted graph. Our algorithm performs an additional step of contracting all infeasible vertices after each contraction step. The description of our algorithm is presented in  Algorithm \ref{fig:alg_node_multi_arb_rank}. 
	We summarize the correctness probability and the run-time of Algorithm \ref{fig:alg_node_multi_arb_rank} in Theorem \ref{thm:nwmomc_arb_rank_suc_prob}. 
	
	\begin{figure*}[ht] 
    \centering\small
    \begin{algorithm} 
    \textul{\textsc{Node-Budgeted-$b$-Multiobjective-Min-Cut-Arbitrary-Rank}($G, t, w, c, b$):}\+
    \\  {\bf Input: } A hypergraph $G=(V,E)$, a positive integer $t$, \+\+ 
    \\ a vector $w$ of vertex-weight functions with $w_i:V\rightarrow \R_+$ for $i \in [t]$, 
    \\ a cost function $c: E \rightarrow \R_+$, and a budget-vector $b \in \mathbb{R}_+^{t}$\-\-
    \\ \rule{0pt}{3ex}$U \gets \{v \in V \colon v \text{ is feasible}\}$
    \\ If $U = \emptyset$:\+
    \\ Return INFEASIBLE \-
    \\ Compute $\alpha_e := \frac{|U \setminus e|}{|U|} \cdot c(e)$ for every hyperedge $e \in E$ 
    \\ If $\alpha_e = 0$ for every hyperedge $e \in E$: \+
    \\ If $U$ is feasible: \+
    \\ Return $\delta(U)$ \-
    \\ Return $E$ \-
    \\ Contract together all infeasible vertices in $G$
    \\ If $U$ is feasible: \+
    \\ Return \textsc{Hypergraph-Min-Cut}($G,c$) \-
    \\ $e \gets$ a random hyperedge of $G$ chosen with probability proportional to $\alpha_e$
    \\ Return \textsc{Node-Budgeted-$b$-Multiobjective-Min-Cut-Arbitrary-Rank}($G, t, w, c, b$) \-
    \end{algorithm}
    \caption{\textsc{Node-Budgeted $b$-Multiobjective Min-Cut} in arbitrary-rank hypergraphs.}
    \label{fig:alg_node_multi_arb_rank}
\end{figure*}

	\begin{theorem}\label{thm:nwmomc_arb_rank_suc_prob} Let $G$ be an n-vertex hypergraph, for some $n \geq 2$ with vertex-weight functions $w_1, \dots, w_{t}: V \rightarrow \mathbb{R}_+$, cost function $c \colon E \rightarrow \mathbb{R}_+$ and  budget-vector $b \in \mathbb{R}^t_+$. Then Algorithm \ref{fig:alg_node_multi_arb_rank} outputs a fixed node-budgeted $b$-multiobjective min-cut in $G$ with probability at least
	\[
	Q_n := \begin{cases} 1 &\text{if } n=2, \\ \frac{1}{3}{n-1 \choose 2}^{-1} &\text{if } n \geq 3. \end{cases}
	\]
	Moreover, the algorithm can be implemented to run in polynomial time.
	\end{theorem}
	
	\begin{proof}
	We first analyze the run-time. Each recursive call in the algorithm reduces the number of vertices, so the total number of recursive calls is at most $n$. Apart from the recursion, the algorithm, verifies the feasibility of each vertex and of $U$, computes $\alpha_e$ for each hyperedge, and either performs a contraction or calls the algorithm for the ordinary hypergraph min-cut problem. All of these can be done in polynomial time.
	
	Now we analyze the correctness probability. Let $\mathcal{G}_n$ be the set of all tuples $(G,t, w, c, b)$ where $G$ is an $n$-vertex hypergraph, $t \in \mathbb{Z}_+$, $w_1, \dots, w_{t}: V(G) \rightarrow \mathbb{R}_+$, $c: E(G) \rightarrow \mathbb{R}_+$, and $b \in \mathbb{R}^t_+$. That is, $\mathcal{G}_n$ is the set of all valid inputs to Algorithm \ref{fig:alg_node_multi_arb_rank} where the hypergraph has $n$ vertices. 
	For an input tuple $T$ in the form just described, let $M(T)$ be the collection of $b$-multiobjective min-cuts for the input instance. Define 
	\[
	q_n := \min_{T \in \mathcal{G}_n} \min_{F \in M(T)} \Pr[\text{Algorithm} \text{ returns } F \text{ on input } T].
	\]
	We will show that $q_n \geq Q_n$ for all $n \geq 2$. We proceed by induction on $n$. As a base case, when $n=2$, we have $\alpha_e=0$ for every $e$, and the algorithm outputs the unique cut with probability 1, so $q_2 = 1 = Q_2$.
	
	We now show the induction step. Let $G$ be a hypergraph on $n \geq 3$ vertices with associated costs, weights, and budgets, and let $F$ be a $b$-multiobjective min-cut in $G$. Assume that $q_{n'} \geq Q_{n'}$ for $2 \leq n' < n$. We will show that the algorithm returns $F$ with probability at least $Q_n = \frac{1}{3}{n-1 \choose 2}^{-1}$.
	
    Suppose $\alpha_e=0$ for every $e\in E$. This means that every hyperedge contains all of the feasible vertices. Let $\emptyset \subsetneq X \subsetneq V$ be a feasible set (one which does not violate the budgets). Then, every vertex in $X$ must be feasible. Since every hyperedge contains all feasible vertices, $\delta(X)$ will either be all of the hyperedges (if $X$ does not contain all feasible vertices) or all hyperedges which contain infeasible vertices (if $X$ contains all feasible vertices). The latter is cheaper, so if the set $U$ of all feasible vertices is still feasible, then we must have $F = \delta(U)$, and the algorithm always returns $F$. Otherwise, every feasible cut contains all hyperedges, so $F = E$ and again the algorithm always returns $F$. We hereafter assume that $\alpha_e > 0$ for some hyperedge $e$.
	
	We note that if $G$ has multiple infeasible vertices, the algorithm will contract them together to yield a hypergraph $G'$ with $n' < n$ vertices and only one infeasible vertex. The probability that the algorithm returns $F$ on input $G$ will be the same as the probability that the algorithm returns $F$ on input $G'$. From our induction hypothesis we know that this probability is at least $Q_{n'} \geq Q_n$. We hereafter assume that $G$ has at most one infeasible vertex.
	
	Next we consider the case where the set $U$ is feasible. (We emphasize that although $w_i(v) \leq b_i$ for every $v \in U$, $i \in [t]$, by the definition of feasible vertices, it need not be the case that $\sum_{v \in U} w_i(v) \leq b_i$ for every $i \in [t]$. So this case does not occur always). Since our vertex weights are all non-negative, if $U$ is feasible, then every subset of $U$ must be feasible as well. Any cut can be written as $\delta(X) = \delta(\complement{X})$ for some $X \subseteq V$. Since $G$ has at most one infeasible vertex, either $X$ or $\complement{X}$ must be a subset of $U$. This means that every cut in $G$ must be feasible. Thus, in this case, the budgets are irrelevant and finding a node-budgeted $b$-multiobjective min-cut is the same as just finding an ordinary minimum cut with respect to the cost function $c$.
	In particular, this means that $F$ is not  only a node-budgeted $b$-multiobjective min-cut in $G$, but it is also a regular min-cut as well. Therefore by Theorem \ref{thm:arb_rank_min_cut_edge_weights}, the algorithm for \textsc{Hypergraph Min-Cut} (i.e., Algorithm \ref{fig:alg_hyper_min_cut}) outputs $F$ with probability at least ${n \choose 2}^{-1}$. Consequently, Algorithm \ref{fig:alg_node_multi_arb_rank} outputs $F$ with probability at least $\binom{n}{2}^{-1}$. Since $n\ge 3$, we have that ${n \choose 2}^{-1}\ge \frac{1}{3}{n-1 \choose 2}^{-1}=Q_n$. 
	
	Finally, suppose that $G$ has at most one infeasible vertex and that $U$ is not feasible in $G$. Then, the algorithm contracts a hyperedge with probability proportional to $\alpha_e$. Let $e'$ be a random variable for the contracted hyperedge. Using the induction hypothesis, we obtain that
	\begin{align*}
	    \Pr[\text{Algorithm} \text{ returns } F \text{ on input } G] 
	    =& \sum_{e \in E \setminus F} \Pr[e' = e] \cdot \Pr[\text{Algorithm} \text{ returns } F \text{ on input } G / e] \\
	    \geq& \sum_{e \in E \setminus F} \frac{\alpha_e}{\sum_{f \in E} \alpha_f} \cdot q_{n-|e|+1} \\
	    \geq& \frac{1}{\sum_{e \in E} \alpha_e} \sum_{e \in E \setminus F} \alpha_e Q_{n-|e|+1}.
    \end{align*}
    Now, Claims \ref{claim:alphaQ_bound_node_budget_gen} and \ref{claim:alphas_bound} complete the proof of the theorem. 
    \end{proof}
	    
	    \begin{claim}\label{claim:alphaQ_bound_node_budget_gen}
	    For every hyperedge $e\in E\setminus F$, we have 
	    \[
	    \alpha_eQ_{n-|e|+1} \geq c(e)Q_n.
	    \]
	    \end{claim}
	    \begin{proof}
	    Suppose $|e| = n-1$. Then $Q_{n-|e|+1} = 1$. Since $U$ is not feasible, we know that $F$ must contain every hyperedge that spans $U$. Since $e\in E\setminus F$, it follows that $|U\setminus e|>0$. Therefore, 
	        $\alpha_e \geq \frac{c(e)}{n}$. We conclude that $\alpha_eQ_{n-|e|+1} = \alpha_e \geq \frac{c(e)}{n} \geq \frac{c(e)}{3}{n-1 \choose 2}^{-1} = c(e)Q_n$.
	        
        Next, suppose $|e| < n-1$. Then $Q_{n-|e|+1} = \frac{1}{3}{n-|e| \choose 2}^{-1}$, and we have
	    \begin{align*}
	        \alpha_eQ_{n-|e|+1} &= \frac{|U \setminus e|c(e)}{|U|} \cdot \frac{1}{3}{n-|e| \choose 2}^{-1} \\
	        &\geq \frac{|U|-|e|}{|U|} \cdot \frac{1}{3}{n-|e| \choose 2}^{-1} \cdot c(e) \\
	        &\geq \frac{n-1-|e|}{n-1} \cdot \frac{2}{3(n-|e|)(n-|e|-1)} \cdot c(e) \\
	        &\geq \frac{2}{3(n-1)(n-|e|)} \cdot c(e) \\
	        & \geq \frac{2}{3(n-1)(n-2)} \cdot c(e) \\
	        &= c(e)Q_n. 
	    \end{align*}
	    \end{proof}

	\begin{claim}\label{claim:alphas_bound}
	\[
    \frac{c(E\setminus F)}{\sum_{e\in E}\alpha_e} \ge 1.
	\]
	\end{claim}
	\begin{proof}
	We consider the cut induced by a uniformly random feasible vertex. A hyperedge $e$ belongs to such a cut with probability $\frac{|U \cap e|}{|U|} = 1 - \frac{|U \setminus e|}{|U|}$. Thus, the expected value of such a cut is $\sum_{e \in E} (1 - \frac{|U \setminus e|}{|U|})c(e) = c(E) - \sum_{e \in E} \alpha_e$. Since the value of the cut induced by a random feasible vertex is an upper bound on the value of a node-budgeted $b$-multiobjective min-cut, this means that $c(F) \leq c(E) - \sum_{e \in E} \alpha_e$. Rewriting this inequality gives $\sum_{e \in E} \alpha_e \leq c(E) - c(F) = c(E \setminus F)$, and the desired inequality follows.
	\end{proof}
	
	\section{Size-Constrained \texorpdfstring{Min-$k$-Cut}{Min-k-Cut} in Arbitrary-Rank Hypergraphs}
	\label{sec:size_constrained}
	
	In this section, we consider the problem of finding a minimum cost $k$-cut subject to constant lower bounds on the weights of the partition classes and prove Theorem \ref{thm:size_constrained_success_prob}. Throughout this section, we assume that $k$ is a constant. We focus on the cardinality case (i.e., unit-cost variant) for the sake of simplicity of exposition and mention that our algorithm also extends to arbitrary non-negative hyperedge costs.
	
	We begin by formally defining the terminology. Let $G=(V,E)$ be a hypergraph. For a partition $X = (X_1, \dots, X_k)$ of $V$, we define $\delta(X)$ to be the set of hyperedges that intersect at least two parts of $X$. For a weight function $w: V \rightarrow \Z_+$, we call $(G,w)$ a \emph{vertex-weighted hypergraph}. We now define our main object of study, namely size-constrained minimum cuts.
	
	\begin{definition}
	Let $G=(V,E)$ be a hypergraph, $w:V\rightarrow \Z_+$ be a vertex-weight function, $k\ge 2$ be an integer, and $s\in \Z_+^k$ be a size-vector. A $k$-partition $X$ of $V$ is an \emph{$s$-size-constrained $k$-partition} if $w(X_i) \ge s_i$ for every $i \in [k]$. A set of hyperedges $F \subseteq E$ is an \emph{$s$-size-constrained $k$-cut} if $F = \delta(X)$ for some $s$-size-constrained $k$-partition $X$. An $s$-size-constrained $k$-cut of minimum cardinality is said to be an \emph{$s$-size-constrained min-$k$-cut}.
	\end{definition}
	
	The following is the central problem of interest in this section.
	\begin{problem}{$s$-Size-Constrained Min-$k$-Cut}
	Given: A vertex-weighted hypergraph $(G,w)$, a positive integer $k$, and a size-vector $s\in \Z^k_+$.
	
	Goal: An $s$-size-constrained min-$k$-cut. 
	\end{problem}
	
	We give a random contraction based algorithm for this problem. Given a hypergraph $G=(V,E)$ and a size-constraint vector $s\in \Z_+^k$, let $n = |V|$. We define $\sigma_j:=\sum_{i=1}^{j}s_i$, and 
	\[
	\alpha_e := \frac{{n-|e| \choose \sigma_{k-1}}}{{n \choose \sigma_{k-1}}}.
	\]
	With these definitions, we solve $s$-size-constrained min-$k$-cut using Algorithm \ref{fig:alg_size_constrained}. We prove Theorem \ref{thm:size_constrained_success_prob} using this algorithm. 
	
	\begin{figure*}[ht] 
    \centering\small
    \begin{algorithm} 
    \textul{\textsc{$s$-Size-Constrained-Min-$k$-Cut}($G, k, s$):}\+
    \\  {\bf Input: } An $n$-vertex hypergraph $G=(V,E)$, an integer $k \geq 2$ \+\+ \\ and size-constraint vector $s = (s_1, \dots, s_k) \in \mathbb{Z}_+$\-\-
    \\ \rule{0pt}{3ex}If $n \leq \max\{2\sigma_{k-1}, \sigma_k\}$ \+
    \\ Pick a random $k$-partition $X$ of $V$ and return $\delta(X)$ \-
    \\ $S \gets$ a random subset of $V$ of size $2\sigma_{k-1}$
    \\ $X_i \gets \emptyset$ for $i \in \{1, \dots, k\}$
    \\ For each $v \in S$: \+
    \\ Pick a random integer $i\in \set{1, \ldots, k}$ and add $v$ to $X_i$\-
    \\ $X_k \gets X_k \cup (V \setminus S)$
    \\ $X \gets (X_1, \dots, X_k)$
    \\ If $X$ is a $k$-partition of $V$: \+
    \\ $R \gets \delta(X)$ \-
    \\ Else: \+
    \\ $R \gets E$ \-
    \\ Compute $\alpha_e := \frac{{n-|e| \choose \sigma_{k-1}}}{{n \choose \sigma_{k-1}}}$ for every $e \in E$
    \\ If $\alpha_e = 0$ for every hyperedge $e \in E(G)$: \+
    \\ Return $R$\-
    \\ $e \gets $ a random hyperedge of $G$ chosen with probability proportional to $\alpha_e$
    \\ $R' \gets$ \textsc{$s$-Size-Constrained-Min-$k$-Cut}($G/e, k, s$) 
    \\ Return $R$ with probability $\frac{1}{n}$ and $R'$ with probability $\frac{n-1}{n}$ \-
    \end{algorithm}
    \caption{\textsc{$s$-Size-Constrained Min-$k$-Cut}}
    \label{fig:alg_size_constrained}
\end{figure*}

	\sizeConstrainedSuccessProb*
	\begin{proof}
	We consider Algorithm \ref{fig:alg_size_constrained}. We first analyze its run-time. Each recursive call reduces the number of vertices in the hypergraph. Thus, the algorithm makes at most $n$ recursive calls. Apart from the recursion steps, the algorithm only selects random partitions and performs contractions, both of which can be implemented to run in polynomial time.
	
	Now we analyze the success probability. Let $\mathcal{G}_n$ be the set of all vertex-weighted $n$-vertex hypergraphs which contain an $s$-size-constrained $k$-cut. For a vertex-weighted hypergraph $(G,w)$, let $M(G,w)$ be the set of all $s$-size-constrained min-$k$-cuts in $(G,w)$. 
	Define 
	\[q_n := \min_{(G,w) \in \mathcal{G}_n} \min_{F \in M(G,w)} \Pr[\text{Algorithm} \text{ returns } F \text{ on input } (G,w,k,s)], \text{ and} 
	\]
	\[
	Q_n := 
	\begin{cases} \left( k^{\max\{2\sigma_{k-1},\sigma_k\}} \right)^{-1} &\text{if } n \leq \max\{2\sigma_{k-1}, \sigma_k\}, and \\ 
	\left( k^{\max\{2\sigma_{k-1},\sigma_k\}} n{n \choose 2\sigma_{k-1}} \right)^{-1} &\text{if } n > \max\{2\sigma_{k-1}, \sigma_k\}.\end{cases}
	\] 
	We note that $Q_n = \Omega(k^{-\sigma_{k-1}-\sigma_k}n^{-2\sigma_{k-1}-1})$, so it suffices to show that $q_n \geq Q_n$ for all $n \geq k$ (for smaller $n$, there are no $k$-cuts).
	
	We proceed by induction on $n$. Let $F$ be an $s$-size-constrained min-$k$-cut in $(G,w)$. Let $Y$ be an $s$-size-constrained $k$-partition with $F = \delta(Y)$. We assume that $|Y_1| \leq |Y_2| \leq \dots \leq |Y_k|$. This assumption is without loss of generality because we can relabel the parts of an $s$-size-constrained $k$-partition so that they are in increasing order of size and the resulting partition will still be an $s$-size-constrained $k$-partition (since we have assumed that the size vector $s$ is in increasing order). 
	
	For the base case, suppose $n \leq \max\{2\sigma_{k-1},\sigma_k\}$. In this case, the algorithm returns $\delta(X)$ where $X$ is a $k$-partition of $V$ chosen uniformly at random. Since $F = \delta(Y)$, the probability that $F = \delta(X)$ for the $X$ randomly chosen by the algorithm is at least $\Pr[X = Y]$. The number of $k$-partitions of $V$ is at most $k^n$, the number of ways to assign each of the $n$ vertices to one of $k$ labeled sets. Thus, $\Pr[X=Y] \geq k^{-n}\geq k^{-max(2\sigma_{k-1},\sigma_k)} = Q_n$.
	
	Now we will prove the inductive step. Assume that $n > \max\{2\sigma_{k-1}, \sigma_k\}$. By the inductive hypothesis, we have $q_{n'} \geq Q_{n'}$ for all $n' \in \set{k, \ldots, n-1}$. We will show that $q_n \geq Q_n$.
	
	Suppose $|Y_k|\ge n-2\sigma_{k-1}$. Let $T$ be an arbitrary subset of $Y_k$ of size $n-2\sigma_{k-1}$. Consider the set $S$ chosen by the algorithm. The probability that $S$ is equal to $V-T$ is ${n \choose 2\sigma_{k-1}}^{-1}$. Next, consider the sets $X_i$ created by the algorithm. The probability that $X_i=Y_i$ for every $i\in [k]$ conditioned on $S = V-T$, is $k^{-2\sigma_{k-1}}$. Thus, the probability that the $k$-partition $X$ obtained in the algorithm is identical to $Y$ is at least $k^{-2\sigma_{k-1}}{n \choose 2\sigma_{k-1}}^{-1}$. Since the last step of the algorithm returns $R=\delta(X)$ with probability $1/n$, it follows that the algorithm returns $F$ with probability at least $\left( nk^{\max\{2\sigma_{k-1},\sigma_k\}}{n \choose 2\sigma_{k-1}} \right)^{-1} = Q_n$.
	
	Henceforth, we assume that $|Y_k|<n-2\sigma_{k-1}$. We will call a hyperedge \emph{large} if it contains at least $n-2\sigma_{k-1}$ vertices. Since $Y_k$ is the largest part of the $k$-partition $Y$, every large hyperedge must be contained in the $k$-cut $F$. In particular, if $\alpha_e = 0$ for a hyperedge $e$, then ${n-|e| \choose \sigma_{k-1}} = 0$, which implies that $n-|e| < \sigma_{k-1}$, and hence $e$ is large and consequently, $e$ cannot be in $F$.
	
	Next, suppose that $\alpha_e=0$ for every hyperedge $e$. Then every hyperedge is a large hyperedge and therefore, $F=E$. In this case, the algorithm will return $R$. We note that $R=E$ if $X$ is not a $k$-partition. We lower bound the probability that $X$ is not a $k$-partition now. If all vertices in $S$ are assigned to $X_k$, then $X$ is not a $k$-partition. The probability that all vertices in $S$ are assigned to $X_k$ is $k^{-2\sigma_{k-1}}$. Thus, the probability that the algorithm returns $F=R=E$ is at least $k^{-2\sigma_{k-1}}\ge Q_n$. 
	
	Henceforth, we assume that $\alpha_e>0$ for some hyperedge $e\in E$. This means that the algorithm will contract some hyperedge and then recurse on the resulting hypergraph. Let $e'$ be a random variable for the hyperedge chosen to be contracted. Let $w'$ be the weight function defined on the vertices of $G/e'$ as follows: $w'(v) := w(v)$ for each $v \in V \setminus e'$ and $w'(v) := \sum_{u \in e'} w(u)$ when $v$ is the new vertex resulting from the contraction. If $e' \not\in F$, then $F$ will be an $s$-size-constrained min-$k$-cut in $(G/e', w')$.
	Therefore, we have that 
	\begin{align*}
	    \Pr[R'=F \text{ on input } (G,k,s)] 
	    &= \sum_{e \in E \setminus F} \Pr[e' = e] \cdot \Pr[\text{Algorithm} \text{ returns } F \text{ on input } (G/e,k,s)]  \\
	    &\ge \sum_{e \in E \setminus F} \frac{\alpha_e}{\sum_{f \in E} \alpha_f} \cdot q_{n-|e|+1} .
	    \end{align*}
	    Let $e$ be a hyperedge that is not in the $k$-cut $F$. Then, $e$ cannot be a large hyperedge and hence, $|e|<n-2\sigma_{k-1}$. Consequently, $n-|e|+1>2\sigma_{k-1}+1\ge k$. Therefore, by applying the induction hypothesis, we have $q_{n-|e|+1}\ge Q_{n-|e|+1}$. Hence,
	    
	    \begin{equation*}
	        \Pr[R'=F \text{ on input } (G,k,s)] \ge\frac{1}{\sum_{f \in E} \alpha_f} \sum_{e \in E \setminus F} \alpha_e \cdot Q_{n-|e|+1}.
	    \end{equation*}
	
	We need the following two claims. We defer their proofs to complete the proof of the theorem. 
	\begin{claim}\label{claim:alphaQ_bound}
	For every hyperedge $e\in E\setminus F$, we have 
	$\alpha_eQ_{n-|e|+1} \geq \frac{nQ_n}{n-1}$. 
	\end{claim}
	\begin{claim} \label{claim:eminusf_geq_alphas}
	$\frac{|E \setminus F|}{\sum_{f \in E} \alpha_f} \geq 1$.
	\end{claim}
	By Claim \ref{claim:alphaQ_bound}, we have that
	\begin{align*}
	     \Pr[R'=F \text{ on input } (G,k,s)] 
	    &=\frac{1}{\sum_{f \in E} \alpha_f} \sum_{e \in E \setminus F} \alpha_e \cdot Q_{n-|e|+1} \\
	    &\geq \frac{1}{\sum_{f \in E} \alpha_f} \sum_{e \in E \setminus F} \frac{nQ_n}{n-1} \\
	    &= \frac{|E \setminus F|}{\sum_{f \in E} \alpha_f} \cdot \frac{nQ_n}{n-1}.
	\end{align*}
	Thus, Claim \ref{claim:eminusf_geq_alphas} implies that
	\[
	\Pr[R'=F \text{ on input } (G,k,s)] \geq  \frac{nQ_n}{n-1}.
	\]
	Finally, we note that since we have assumed $n > \max\{2\sigma_{k-1}, \sigma_k\}$ and $\alpha_e > 0$ for some $e$, the probability that the algorithm returns $R'$ is $(n-1)/n$. Thus, we conclude that $\Pr[ \text{Algorithm} \text{ returns } F] \geq Q_n$.
	\end{proof}

	\begin{proof}[Proof of Claim \ref{claim:alphaQ_bound}]
	Let $e\in E\setminus F$. We recall that $F$ contains all large hyperedges and hence,
	
	\begin{equation}\label{ineq:e_is_small}
	    |e| < n - 2\sigma_{k-1}.
	\end{equation}
	
	First, suppose that $n-|e|+1 \leq \max\{2\sigma_{k-1}, \sigma_k\}$. Then, we have $Q_{n-|e|+1} = k^{-\max\{2\sigma_{k-1},\sigma_k\}}$. We consider two subcases. 
	
	\begin{casesp}
	\item Suppose $n \geq 3\sigma_{k-1}$. Then
	
	\begin{equation*}
	    \frac{{n-|e| \choose \sigma_{k-1}}}{{n \choose \sigma_{k-1}}} Q_{n-|e|+1} \geq \frac{Q_{n-|e|+1}}{{n \choose \sigma_{k-1}}} \geq \left( k^{\max\{2\sigma_{k-1},\sigma_k\}}{n \choose 2\sigma_{k-1}} \right)^{-1} = nQ_n \geq \frac{nQ_n}{n-|e|+1}.
	\end{equation*}
	The first inequality follows from inequality (\ref{ineq:e_is_small}), and the last inequality follows from the fact that ${n \choose \sigma_{k-1}} \leq {n \choose 2\sigma_{k-1}}$ for $n \geq 3\sigma_{k-1}$.
	
	\item Suppose $n < 3\sigma_{k-1}$. We recall that $n>2\sigma_{k-1}$. By inequality (\ref{ineq:e_is_small}), $n-|e| > 2\sigma_{k-1}$. Letting $n=2\sigma_{k-1}+x$ for some $x \in \set{1, \dots \sigma_{k-1}}$, we have
	
	\begin{equation*}
	    \frac{{n-|e| \choose \sigma_{k-1}}}{{n \choose \sigma_{k-1}}} \geq \frac{{2\sigma_{k-1} \choose \sigma_{k-1}}}{{2\sigma_{k-1}+x \choose \sigma_{k-1}}}
	    = \frac{(2\sigma_{k-1})!(\sigma_{k-1}+x)!}{(2\sigma_{k-1}+x)!\sigma_{k-1}!} = \prod_{i=1}^x \frac{\sigma_{k-1}+i}{2\sigma_{k-1}+i} \geq \left( \frac{1}{2} \right)^x.
	\end{equation*}
	We also know that
	\begin{equation*}
	    {n \choose 2\sigma_{k-1}}^{-1} = {2\sigma_{k-1}+x \choose 2\sigma_{k-1}}^{-1} = \frac{(2\sigma_{k-1})!x!}{(2\sigma_{k-1}+x)!}
	    = \prod_{i=1}^x \frac{i}{2\sigma_{k-1}+i} \leq \left( \frac{1}{2} \right)^x.
	\end{equation*}
	Therefore,
	\begin{equation*}
	    \frac{{n-|e| \choose \sigma_{k-1}}}{{n \choose \sigma_{k-1}}} Q_{n-|e|+1} \geq \left( \frac{1}{2} \right)^x Q_{n-|e|+1} \geq \frac{Q_{n-|e|+1}}{{n \choose 2\sigma_{k-1}}} = \left( k^{\max\{2\sigma_{k-1},\sigma_k\}}{n \choose 2\sigma_{k-1}} \right)^{-1} = nQ_n \geq \frac{nQ_n}{n-|e|+1}.
	\end{equation*}
	\end{casesp}
	
	Next, suppose that $n-|e|+1>\max\{2\sigma_{k-1}, \sigma_k\}$. We have that
	\begin{align*}
	    \alpha_e Q_{n-|e|+1} &= \frac{{n-|e| \choose \sigma_{k-1}}}{{n \choose \sigma_{k-1}}}Q_{n-|e|+1} \\
	    &=\frac{{n-|e| \choose \sigma_{k-1}}}{{n \choose \sigma_{k-1}}} \cdot \frac{1}{{n-|e|+1 \choose 2\sigma_{k-1}}} \cdot \left( (n-|e|+1)k^{\max\{2\sigma_{k-1},\sigma_k\}} \right)^{-1} \\
	    &= \frac{{n-|e| \choose \sigma_{k-1}}}{{n \choose \sigma_{k-1}}} \cdot \frac{1}{{n-|e|+1 \choose 2\sigma_{k-1}}} \cdot \frac{n}{n-|e|+1}{n \choose 2\sigma_{k-1}} \cdot Q_n \\
	    &\geq \frac{{n-|e| \choose \sigma_{k-1}}}{{n \choose \sigma_{k-1}}} \cdot \frac{1}{{n-|e|+1 \choose 2\sigma_{k-1}}} \cdot {n \choose 2\sigma_{k-1}} \cdot \frac{nQ_n}{n-1}.
	 \end{align*}
	  The following proposition completes the proof. We defer its proof to the appendix.
	  
	  \begin{restatable}{proposition}{ratioIneq}\label{prop:ratio_ineq}
	  For positive integers $n, e, \sigma$ with $e\ge 2$ and $n-e+1>2\sigma$, we have 
	\[
	\frac{{n-e \choose \sigma}}{{n \choose \sigma}} \cdot \frac{1}{{n-e+1 \choose 2\sigma}} \geq {n \choose 2\sigma}^{-1}.
	\]
	\end{restatable}
	\end{proof}
	
	\begin{proof}[Proof of Claim \ref{claim:eminusf_geq_alphas}]
	Let $Z = (Z_1, \dots, Z_k)$ be a random $k$-partition obtained by picking disjoint sets $Z_1, \dots, Z_{k-1}$ with $|Z_i| = s_i$ and setting $Z_k = V \setminus \bigcup_{i=1}^{k-1} Z_i$. Since $n > \sigma_k$ and every vertex has weight at least $1$, the $k$-partition $Z$ is an $s$-size-constrained $k$-partition. Therefore, $|\delta(Z)|$ is an upper bound on $|F|$. In particular, 
	\[
	|F| \leq \mathbb{E}(|\delta(Z)|) = \sum_{e \in E} \Pr(e \in \delta(Z)).
	\]
	Negating the inequality and adding $|E|$ to both sides gives
	\begin{align*}
	    |E \setminus F| &\geq \sum_{e \in E} (1 - \Pr(e \in \delta(Z))) \\
	    &= \sum_{e \in E} \Pr(e \not\in \delta(Z)) \\
	    &= \sum_{e \in E} \sum_{i = 1}^k \Pr(e \subseteq Z_i) \\
	    &\geq \sum_{e \in E} \frac{{n-|e| \choose \sigma_{k-1}}}{{n \choose \sigma_{k-1}}} \\
	    &= \sum_{e \in E} \alpha_e.
	\end{align*}
	Thus, $|E \setminus F| / \sum_{f \in E} \alpha_f \geq 1$.
	\end{proof}
	
	\noindent \textbf{Remark.} Since our algorithm does not even take the vertex-weights as input, it could trivially be extended to handle a version of the problem where we have multiple weight functions on the vertices (as in the previous sections) each with their own minimum sizes. If we have $t$ vertex-weight functions, $w_1, \dots, w_{t} : V \rightarrow \mathbb{Z}_+$ and each function $w_j$ has an associated list of lower bounds $s_{j,1}, \dots, s_{j,k}$, then we can find a min-$k$-cut satisfying all of these lower-bound constraints with at least inverse polynomial probability by simply running our algorithm with $s_i = \max_{j \in [t]} s_{j,i}$ for every $i\in [t]$.
	
\section{Conclusion and Open Problems}
In this work, we illustrated the versatility of the random contraction technique by addressing multicriteria versions of min-cut and size-constrained min-$k$-cut problems. There are several interesting open questions in this area. We conclude by stating a few: 
(1) For the number of pareto-optimal cuts and multiobjective min-cuts, there is still a gap between our lower bound (which is $\Omega(n^{t})$) and our upper bound (which is $O(n^{3t-1})$). 
Can we improve either of these bounds? 
We believe that improving our bounds for the number of $b$-multiobjective min-cuts for a fixed budget-vector $b\in \R^{t-1}_+$ would be a first-step towards this goal. 
(2) We gave a polynomial-time algorithm to solve the $b$-multiobjective min-cut problem in constant-rank hypergraphs. How about arbitrary-rank hypergraphs? Is the $b$-multiobjective min-cut problem in arbitrary rank hypergraphs (even for $t=2$ criteria) solvable in polynomial-time or is it NP-hard? 
	
\bibliographystyle{plainurl}
\bibliography{references} 

\renewcommand{\figurename}{Figure}

\appendix 
\section{Appendix}

\subsection{Comparison of Parametric, Pareto-Optimal, and Multiobjective Cuts}\label{sec:comparison}
We prove the containment relationship (\ref{containment-relationships}) here. 
\begin{proposition} \label{prop:containment-relationships}
The following containment relationship holds, possibly with the containment being strict: 
\begin{align*} 
\text{Parametric min-cuts} &\subseteq \text{Pareto-optimal cuts} \subseteq \text{Multiobjective
min-cuts}.
\end{align*}
\end{proposition}

	\begin{proof}
	We first show that parametric min-cuts are pareto-optimal cuts: If a cut $F'$ dominates a cut $F$, then $w(F') < w(F)$ for all positive multipliers, and therefore $F$ cannot be a parametric min-cut. On the other hand, not every pareto-optimal cut is a parametric min-cut (see Figure \ref{fig:pareto-is-not-parametric} for an example).
	
	Next, we show that pareto-optimal cuts are multiobjective min-cuts: If a cut $F$ is pareto-optimal, then it is a $b$-multiobjective min-cut for the budget-vector $b$ obtained by setting $b_i:=c_i(F)$ for every $i\in [k-1]$. On the other hand, not every multiobjective min-cut is a pareto-optimal cut (see Figure \ref{fig:multi-is-not-pareto} for an example). 
	
\begin{figure}[ht]
\centering
\begin{minipage}{.5\textwidth}
  \centering
  \includegraphics[trim={0 12cm 24cm 0.5cm}, clip, scale=0.5]{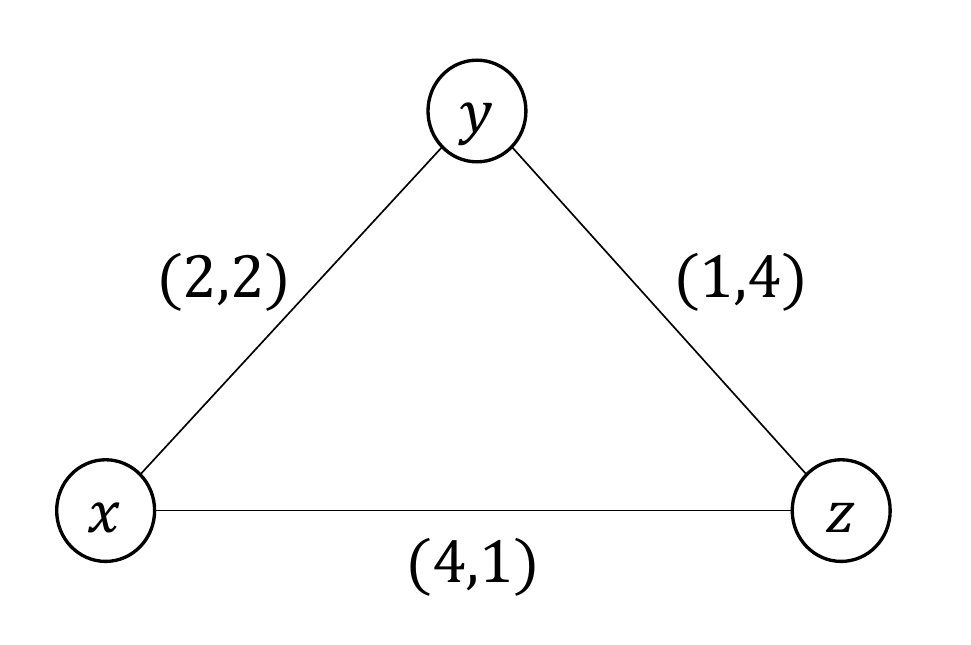}
	    \caption{
	    The cut $\delta(z)$ is a pareto-optimal cut but not a parametric min-cut.}
	    \label{fig:pareto-is-not-parametric}
\end{minipage}%
\begin{minipage}{.5\textwidth}
  \centering
\includegraphics[trim={0 16cm 24cm 1cm}, clip, scale=0.5]{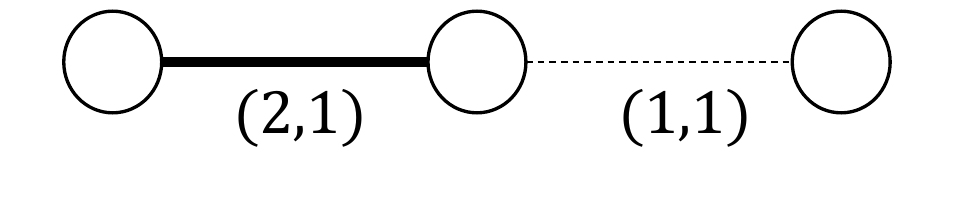}
	    \caption{
	    For $b=2$, the bold edge is a $b$-multiobjective min-cut but it is not a pareto-optimal cut. }
	    \label{fig:multi-is-not-pareto}
\end{minipage}
\end{figure}
	\end{proof}

\subsection{Proof of Lemma \ref{lem:lp_sol}}
We restate and prove Lemma \ref{lem:lp_sol}.

\LPSol*

	\begin{proof}
	The linear program is feasible, since setting $x_2=1$ and the rest of the variables to zero gives a feasible solution. Let $j \in \{2, \dots, r\}$. Since $y_j \geq 0$ and $x_j \leq 1$, we have that $x_j - y_j \leq 1$. Since $f(n-j+1) \ge 0$ for every $j\in [2,r]$, it follows that $(x_j-y_j)f(n-j+1) \leq f(n-j+1)$. Therefore we have $\sum_{j=2}^r (x_j-y_j)f(n-j+1) \leq \sum_{j=2}^r f(n-j+1)$. Therefore, the objective value of this linear program is bounded. Since the linear program is feasible and bounded, there exists an extreme point optimum solution to this LP. Since the LP has $2r-2$ variables and $2r$ equations, every extreme point optimum will have at least $2r-2$ tight constraints and at most $2$ non-tight constraints.
		
	We now show that constraint (\ref{con:tys_leq_xs}) is tight for every optimal solution $(x,y)$. 
	Let $(x,y)$ be an optimal solution. Since $\gamma \geq r+1$,  we have 
	\[\gamma \sum_{j=2}^r y_j \geq \sum_{j=2}^r (r+1)y_j > \sum_{j=2}^r j \cdot y_j. \]
	This implies that we cannot have $y_j = x_j$ for all $j$, otherwise $(x,y)$ would violate constraint (\ref{con:tys_leq_xs}). Hence, at least one of the $y_j \leq x_j$ constraints must be slack. Let $j\in \{2, \dots, r\}$ be such that $y_j<x_j$. If constraint (\ref{con:tys_leq_xs}) were slack, increasing the value of $y_j$ by a very small amount would improve the objective value of $(x,y)$ without violating any constraints. Therefore, since $(x,y)$ is optimal, constraint (\ref{con:tys_leq_xs}) must be tight.
		
		Let $(x,y)$ be an extreme point optimal solution. Since we know that $\sum_{j=2}^r x_j = 1$ and $x_j \geq 0$ for every $j \in \{2,\dots,r\}$, we have $\sum_{j=2}^r j \cdot x_j > 0$. Since constraint (\ref{con:tys_leq_xs}) is tight for $(x,y)$, we must have $\gamma \sum_{j=2}^r y_j > 0$. This implies that there exists $j\in [2,r]$ such that $y_j>0$.
		Thus, we conclude that the two slack constraints must be $0 \leq y_{j_1}$ and $y_{j_2} \leq x_{j_2}$ for some $j_1, j_2\in \{2,\ldots, r\}$. We consider two cases.
		
		\begin{casesp}
			\item Suppose $j_1 = j_2$. Then we have that $0 < y_j < x_j = 1$ for some $j \in \{2,\ldots, r\}$, and $x_{j'},y_{j'}=0$ for every $j'\in \{2,\ldots, r\}\setminus \{j\}$. Therefore, we can simplify our LP to 
			\begin{equation*}
			\begin{aligned}
			& \underset{y_j}{\text{minimize}}
			& & (1-y_j)f(n-j+1) \\
			& \text{subject to}
			& & 0 \leq y_j \leq 1 \\
			&&& \gamma y_j = j.
			\end{aligned}
			\end{equation*}
			The only (and therefore optimal)  solution to this LP is $y_j = \frac{j}{\gamma}$, which achieves an objective value of \[\left(1-\frac{j}{\gamma}\right)f(n-j+1).\]
			
			\item Suppose $j_1 \neq j_2$. Then we have that $0 < y_{j_1} = x_{j_1}$, and $0 = y_{j_2} < x_{j_2}$. We note that $x_{j_2} = 1-x_{j_1}$, and therefore we can simplify the LP to
			\begin{equation*}
			\begin{aligned}
			& \underset{x_{j_1}}{\text{minimize}}
			& & (1-x_{j_1})f(n-j_2+1) \\
			& \text{subject to}
			& & 0 \leq x_{j_1} \leq 1 \\
			&&& \gamma x_{j_1} = j_1 \cdot x_{j_1} + j_2 \cdot (1-x_{j_1}).
			\end{aligned}
			\end{equation*}
			Solving the second constraint for $x_{j_1}$ yields $x_{j_1} = \frac{j_2}{\gamma-j_1+j_2} $, and therefore our optimal objective value is  
			\[{\left(1- \frac{j_2}{\gamma-j_1+j_2}\right)f(n-j_2+1}).\]
			
		\end{casesp}
		
		We conclude that the optimal objective value of the LP is equal to the minimum of the values from these two cases, that is, 
		\[\min\left\{\min_{j\in \{2,\ldots, r\}} \left\{\left(1-\frac{j}{\gamma}\right)f(n-j+1)\right\}, \, \, \, \min_{j_1,j_2 \in \{2,\ldots, r\}} \left\{\left(1- \frac{j_2}{\gamma-j_1+j_2})f(n-j_2+1)\right)\right\}\right\}.
		\]
		Since $(1- \frac{j_2}{\gamma-j_1+j_2})$ is decreasing in $j_1$ and $f(n-j_2+1)$ is always positive, we have 
		\[\min_{j_1,j_2\in\{2,\ldots,r\}} \left\{\left(1- \frac{j_2}{\gamma-j_1+j_2}\right)f(n-j_2+1)\right\} = \min_{j\in \{2,\ldots, r\}} \left\{\left(1-\frac{j}{\gamma-r+j}\right)f(n-j+1)\right\}.\] 
		Thus, since $j \leq r$, the optimal  objective value of the LP is equal to
		\begin{equation*}
		 \min_{j\in \{2,\ldots, r\}}\left\{\min\left\{1-\frac{j}{\gamma}, 1-\frac{j}{\gamma-r+j}\right\}f(n-j+1)\right\} 
		= \, \min_{j\in \{2,\ldots, r\}}\left\{\left(1-\frac{j}{\gamma-r+j}\right)f(n-j+1)\right\}.
		\end{equation*}
	\end{proof}

\subsection{Proof of Proposition \ref{prop:ratio_ineq}}
We restate and prove Proposition \ref{prop:ratio_ineq}.
	\ratioIneq*
	
	\begin{proof}
	We note that
	\begin{align}
	    \frac{{n-e \choose \sigma}}{{n \choose \sigma}} \cdot \frac{1}{{n-e+1 \choose 2\sigma}}
	    &= \frac{(n-e)!(n-\sigma)!}{n!(n-e-\sigma)!} \cdot \frac{(2\sigma)!(n-e-2\sigma+1)!}{(n-e+1)!} \nonumber \\
	    &= \left( \prod_{i=0}^{\sigma-1} \frac{n-e-i }{ n-i } \right) \cdot \frac{(2\sigma)!}{\prod_{i=0}^{2\sigma-1} (n-e+1-i)}. \label{exp:expanded_binomial_ratio}
	\end{align}
	To lower bound this expression, we case on the value of $e$.
	
	\begin{casesp}
	\item Suppose $e > \sigma$. Then we can lower bound expression (\ref{exp:expanded_binomial_ratio}) by
	\begin{align*}
	    \frac{1}{\prod_{i=0}^{\sigma-1} (n-i)} \cdot \frac{(2\sigma)!}{(n-e+1) \prod_{i=0}^{\sigma-2} (n-e-\sigma-i)} 
	    \geq& \frac{(2\sigma)!}{\prod_{i=0}^{2\sigma-1} (n-i)} 
	    = {n \choose 2\sigma}^{-1}.
	\end{align*}
	
	\item Suppose $e \leq \sigma$.  We note that
	\begin{equation*}
	     \frac{(2\sigma)!}{\prod\limits_{i=0}^{2\sigma-1} (n-e+1-i)} 
	    =  \frac{(2\sigma)!}{\prod\limits_{i=0}^{2\sigma-1} (n-i)} \cdot \prod_{i=0}^{e-2} \frac{n-i}{n-2\sigma-i}
	    ={n \choose 2\sigma}^{-1} \cdot \prod_{i=0}^{e-2} \frac{n-i}{n-2\sigma-i}.
	\end{equation*}
	Thus, expression (\ref{exp:expanded_binomial_ratio})  is equal to 
	\begin{equation*}
	     {n \choose 2\sigma}^{-1} \cdot \left( \prod_{i=0}^{\sigma-1} \frac{n-e-i }{ n-i } \right) \left( \prod_{i=0}^{e-2} \frac{n-i}{n-2\sigma-i} \right).
	\end{equation*}
	
	We will show that $\left( \prod_{i=0}^{\sigma-1} \frac{n-e-i }{ n-i } \right) \left( \prod_{i=0}^{e-2} \frac{n-i}{n-2\sigma-i} \right) \geq 1$. We note that
	\begin{align}
	    \left( \prod_{i=0}^{\sigma-1} \frac{n-e-i }{ n-i } \right) \left( \prod_{i=0}^{e-2} \frac{n-i}{n-2\sigma-i} \right) &= \frac{\prod_{i=0}^{\sigma-1} (n-e-i) }{\prod_{i=e-1}^{\sigma-1} (n-i) } \cdot \frac{1 }{\prod_{i=0}^{e-2} (n-2\sigma-i) } \nonumber \\
	    &= \frac{\prod_{i=\sigma}^{e+\sigma-1} (n-i) }
	    {(n-e+1)\prod_{i=0}^{e-2} (n-2\sigma-i) }. \label{exp:prod_ratio}
	\end{align}
	We claim that expression (\ref{exp:prod_ratio}) is minimized when $e = 2$. To see this, we note that
	\begin{equation*}
	    \frac{\prod\limits_{i=\sigma}^{(e+1)+\sigma-1} (n-i) }
	    {(n-(e+1)+1)\prod\limits_{i=0}^{(e+1)-2} (n-2\sigma-i) } = \frac{\prod\limits_{i=\sigma}^{e+\sigma-1} (n-i) }
	    {(n-e+1)\prod\limits_{i=0}^{e-2} (n-2\sigma-i) } \cdot \frac{(n-e-\sigma)(n-e+1) }{(n-2\sigma-e+1)(n-e)}.
	\end{equation*}
	Since $e \leq \sigma$, we know that $n-e-\sigma \geq n-2\sigma+1$. From this, along with the fact that $n-e+1 > n-e$, we conclude that $\frac{(n-e-\sigma)(n-e+1) }{(n-2\sigma-e+1)(n-e)} > 1$. This means that expression (\ref{exp:prod_ratio}) increases when we increment $e$. Thus
	\begin{equation*}
	    \frac{\prod_{i=\sigma}^{e+\sigma-1} (n-i) }
	    {(n-e+1)\prod_{i=0}^{e-2} (n-2\sigma-i) }
	    \geq
	    \frac{(n-\sigma)(n-\sigma-1)}{(n-1)(n-2\sigma)} 
	    = \frac{n^2 - (2\sigma+1)n + (\sigma+1)\sigma}{n^2 - (2\sigma+1)n +2\sigma } 
	    \geq 1.
	\end{equation*}
	The last inequality follows from the fact that $\sigma \geq 1$.
	
	Thus, we have shown that $\left( \prod_{i=0}^{\sigma-1} \frac{n-e-i }{ n-i } \right) \left( \prod_{i=0}^{e-2} \frac{n-i}{n-2\sigma-i} \right) \geq 1$, and therefore, combining the above inequalities, we have that
	\begin{equation*}
	    \frac{{n-e \choose \sigma}}{{n \choose \sigma}} \cdot \frac{1}{{n-e+1 \choose 2\sigma}} \geq {n \choose 2\sigma}^{-1}.
	\end{equation*}
	\end{casesp}
	\end{proof}

\end{document}